\documentclass[gpscopy,onehalfspacing,11pt]{ubcdiss}
\usepackage[T1]{fontenc}
\usepackage{amsmath}
\usepackage{amssymb}

\usepackage[colorinlistoftodos]{todonotes}

\usepackage[english]{babel}
\usepackage[utf8]{inputenc}
\usepackage{mdwlist}
\usepackage{verbatim}
\usepackage{accents}
\usepackage{subcaption}
\usepackage{placeins}

\usepackage{comment}

\usepackage{tikz}
\usetikzlibrary{calc}
\usetikzlibrary{arrows}
\usetikzlibrary{decorations.pathreplacing,decorations.markings}

\newtheorem{theorem}{Theorem}

\newtheorem{definition}[theorem]{Definition}

\newenvironment{proof}[1][Proof]{\noindent\textbf{#1.} }{\ \rule{0.5em}{0.5em}}

\newcommand{\bra}[1]{\langle#1|} \newcommand{\ket}[1]{|#1\rangle}
\newcommand{\ketbra}[2]{|#1\rangle\!\langle#2|}
\newcommand{\braket}[2]{\langle#1|#2\rangle}

\newcommand{\tr}{\text{tr}}
\newcommand{\rank}{\text{rank}}

\tikzset{
    >=stealth',
    punkt/.style={
           rectangle,
           rounded corners,
           draw=black, very thick,
           text width=6.5em,
           minimum height=2em,
           text centered},
    pil/.style={
           ->,
           thick,
           shorten <=2pt,
           shorten >=2pt,},
  on each segment/.style={
    decorate,
    decoration={
      show path construction,
      moveto code={},
      lineto code={
        \path [#1]
        (\tikzinputsegmentfirst) -- (\tikzinputsegmentlast);
      },
      curveto code={
        \path [#1] (\tikzinputsegmentfirst)
        .. controls
        (\tikzinputsegmentsupporta) and (\tikzinputsegmentsupportb)
        ..
        (\tikzinputsegmentlast);
      },
      closepath code={
        \path [#1]
        (\tikzinputsegmentfirst) -- (\tikzinputsegmentlast);
      },
    },
  },
  mid arrow/.style={postaction={decorate,decoration={
        markings,
        mark=at position .5 with {\arrow[#1]{stealth'}}
      }}}
}

\newcommand\bellbra[2]{
    \draw[fill] (#1,#2) circle [radius=0.1];
    \draw [thick] (#1-2,#2) -- (#1+2,#2);
    \draw[thick,->] (#1,#2) -- (#1-1,#2);
    \draw[thick,->] (#1,#2) -- (#1+1,#2);
}

\newcommand\bellket[2]{
    \draw[fill] (#1,#2) circle [radius=0.1];
    \draw [thick] (#1-2,#2) -- (#1+2,#2);
    \draw[thick,->] (#1-2,#2) -- (#1-1,#2);
    \draw[thick,->] (#1+2,#2) -- (#1+1,#2);
}

\newcommand\arrowedcurve[6]{
	\draw[thick,postaction={on each segment={mid arrow}}] (#1,#2) to [out=#5,in=#6] (#3,#4);
}

\newcommand\bluearrowedsegment[4]{
	\draw[blue,thick,postaction={on each segment={mid arrow}}] (#1,#2) -- (#3,#4);
}

\newcommand\blueketsegment[5]{
\bluearrowedsegment{#1}{#2}{#1/2+#3/2}{#2/2+#4/2}
\bluearrowedsegment{#3}{#4}{#1/2+#3/2}{#2/2+#4/2}
\draw[fill,blue] (#1/2+#3/2,#2/2+#4/2) circle [radius=#5];
}

\newcommand\bluebrasegment[5]{
\bluearrowedsegment{#1/2+#3/2}{#2/2+#4/2}{#1}{#2}
\bluearrowedsegment{#1/2+#3/2}{#2/2+#4/2}{#3}{#4}
\draw[fill,blue] (#1/2+#3/2,#2/2+#4/2) circle [radius=#5];
}

\newcommand\arrowedsegment[4]{
	\draw[thick,postaction={on each segment={mid arrow}}] (#1,#2) -- (#3,#4);
}

\usepackage{times,mathptmx,courier}
\usepackage[scaled=.92]{helvet}

\usepackage{pifont}

\usepackage{checkend}	
\usepackage{graphicx}	

\usepackage{booktabs}

\usepackage{listings}
\usepackage[printonlyused,nohyperlinks]{acronym}

\usepackage{color}
\definecolor{greytext}{gray}{0.5}

\usepackage{comment}

\usepackage[numbers,sort&compress]{natbib}

\usepackage[compact]{titlesec}
\titleformat*{\section}{\singlespacing\raggedright\bfseries\Large}
\titleformat*{\subsection}{\singlespacing\raggedright\bfseries\large}
\titleformat*{\subsubsection}{\singlespacing\raggedright\bfseries}
\titleformat*{\paragraph}{\singlespacing\raggedright\itshape}

\usepackage[format=hang,indention=-1cm,labelfont={bf},margin=1em]{caption}

\usepackage{url}
\usepackage[bookmarks,bookmarksnumbered,%
    allbordercolors={0.8 0.8 0.8},%
    pagebackref,linktocpage%
    ]{hyperref}
\DeclareUrlCommand\DOI{}
\newcommand{\doi}[1]{\href{http://dx.doi.org/#1}{\DOI{doi:#1}}}



\title{Thesis: Tensor networks for dynamic spacetimes}

\author{Alex May}
\previousdegree{Bsc Joint Honours Physics and Mathematics, McGill University, 2014}

\degreetitle{Master of Science}

\institution{The University of British Columbia}
\campus{Vancouver}

\faculty{The Faculty of Graduate and Postdoctoral Studies}
\department{Physics}
\submissionmonth{August}
\submissionyear{2017}

\hypersetup{
  pdftitle={Tensor networks for dynamic spacetimes  (DRAFT: \today)},
  pdfauthor={Alex May},
  pdfkeywords={AdS/CFT, quantum gravity, tensor network, quantum information}
}

\begin{document}


\maketitle

\chapter{Abstract}

Tensor networks give simple representations of complex quantum states. They have proven useful in the study of condensed matter systems and conformal fields, and recently have provided toy models of AdS/CFT. Underlying the tensor network - AdS/CFT connection is the association of a graph geometry with the tensor network. In the context of the AdS/CFT correspondence tensor network models have so far been limited to describing static spacetimes. In this thesis we look to extend tensor network models of AdS/CFT by describing the geometry of a dynamic spacetime using a tensor network. We provide a review of tensor networks in the context of AdS/CFT to motivate this extension, before proposing modifications of holographic tensor network models that capture features of AdS/CFT with dynamic spacetimes.

This thesis includes the results of arXiv submission 1611.06220 \cite{may2017tensor}, along with a review of the relevant literature. 

\cleardoublepage



\renewcommand{\contentsname}{Table of Contents}
\tableofcontents
\cleardoublepage	





\textspacing		


\chapter{Acknowledgments}

This work was carried out under the supervision of Mark Van Raamsdonk, who was involved in discussions throughout this project and reviewed drafts of this manuscript. Charles Rabideau made useful contributions in the early stages of the project. We are also indebted to Michael Walter, Grant Salton, Zhao Yang, David Stephen, Jordan Cotler, and Adam Levine for helpful discussions. Dominik Neuenfeld and Jaehoon Lee provided feedback on early versions of the manuscript.

AM was partially supported by the It from Qubit Collaboration, which is sponsored by the Simons Foundation. AM was also supported by a CGSM award given by the National Research Council of Canada. This research benefited from the It from Qubit summer school held at the Perimeter Institute for Theoretical Physics. Research at Perimeter Institute is supported by the Government of Canada through Industry Canada and by the Province of Ontario through the Ministry of Economic Development \& Innovation.

AM also wishes to acknowledge the personal support of Mark Van Raamsdonk, whose flexible style of direction allowed the author to pursue varied directions of research, Jessica Allanach, for her patience and personal support, the Lertzman-Lepofsky's, whose welcoming home provided the backdrop for finding many of the results presented here, to Patrick Hayden for early guidance and encouragement, and to Fang Xi Lin, Eric Hanson, Joelle and Bernie Ducharme, and both the UBC and McGill physics communities. 


\mainmatter

\acresetall	

\include{intro}


\chapter{Holography and the AdS/CFT correspondence}

Gravity is well described by general relativity, which can be understood as an effective field theory. Leading order quantum corrections to general relativity can be calculated by standard effective field theory methods \cite{burgess2004quantum}. The effective fields approach however requires a high energy cutoff, and this description of gravity is expected to fail at energies near to and above this cutoff. 

One approach to finding a UV complete theory of gravity is to search for an appropriate field theory which is renormalizable and whose low energy description looks like the (non-renormalizable) terms in the Einstein-Hilbert action. Indeed, the description of the weak force was initially in terms of a non-renormalizable Lagrangian, which was later found to have a high energy renormalizable description. 

There are basic reasons however to expect this approach to fail and a more radical proposal to be necessary. One such reason is the understanding that a UV complete theory of quantum gravity must be holographic \cite{hooft1994dimensional,susskind1995world}, which a standard field theory description is not. To understand what a holographic theory is we will briefly review some facts about black holes and quantum mechanics below. 

\section{The holographic principle}

Black holes provide a key low energy probe of high energy physics, and for this reason have long been at the focus of research in quantum gravity. An early important line of work on black holes concerned black hole thermodynamics, which showed that the quantities which describe a black hole - mass, area, charge, and angular momentum - obey relations identical to those of ordinary thermodynamics \cite{wald1994quantum, wald2010general}. In particular, the area of the black hole plays the role of the thermodynamic entropy in this analogy. 

Black hole thermodynamics lead Bekenstein to argue that the similarity between the black hole relations and thermodynamic relations is more than an analogy, and that in fact the area of the black hole really was the entropy of the black hole\cite{bekenstein1972black,bekenstein1973black}. A key argument of Bekensteins is that making the identification 
\begin{align}
S_{BH} = \frac{A}{4G}
\end{align}
would lead to
\begin{align}
\delta S_{total} = \delta(S_{BH} + S_{matter}) \geq 0.
\end{align}
Thus, identifying the black hole area as the entropy of the black hole correctly leads to the second law of thermodynamics being maintained in a black hole spacetime.

Bekensteins claim was put on firm ground by Hawking, who did a calculation in quantum field theory in a black hole background \cite{hawking1975particle} to determine that the black hole would radiate at a temperature
\begin{align}
T = \frac{1}{8\pi G M}.
\end{align}
This temperature is consistent with the thermodynamic relation $dE = T dS$ when the entropy is identified as $S=A/4G$, as argued by Bekenstein.

The simple statements above about black holes are already enough to argue for the holographic principle. Suppose we start with a physical system which is contained in a spherical spatial region $\Gamma$. Call the thermal entropy of this system $S(\Gamma)$. Then it is straightforward to show that
\begin{align}\label{eq:dofentropybound}
\log N \geq S(\Gamma),
\end{align}
where $N$ is the number of states available to the system $\Gamma$. We will suppose that the system has nearly achieved this bound, and has a maximal entropy so that  $S(\gamma) \approx \log N$.

Now suppose that we add matter to our spatial region $\Gamma$ until a black hole forms which occupies all of $\Gamma$. Then the regions entropy is $A / 4G$, as given by the Bekenstein-Hawking formula. The second law of thermodynamics says that the entropy after matter is added should be greater than before, so that
\begin{align}
S(\Gamma) \leq \frac{A}{4G}.
\end{align}
Combining this with the relation $S(\Gamma) \approx \log N$,
\begin{align}
\log N \leq \frac{A}{4G}.
\end{align}
That is, the number of degrees of freedom in a region $\Gamma$ is bounded above by the area of the region. 

This should come as a surprise. A field theory description of a typical system involves degrees of freedom at each spatial point, and thus the total number of degrees of freedom required scales with the volume. The field theory is only a low energy effective theory however, and our black hole arguments indicate that in fact the complete theory should have far fewer degrees of freedom. Motivated by this, we can state the holographic principle as follows: \emph{a UV complete theory of quantum gravity in $3+1$ dimensions should have an equivalent description in $2+1$ dimensions.} 

\section{Quantum information theory}\label{sec:qit}

The holographic principle tells us that a quantum theory of gravity can be written in one fewer spatial dimensions than expected. Because of the success of non-holographic physics, we know that at least low energy processes do have a simple, local Lagrangian description in $3+1$ dimensions. These processes should have equivalent descriptions in the lower dimensional theory, which already raises basic questions. Since naively a local quantum field theory in $3+1$ dimensions is described by a larger dimensional Hilbert space than a $2+1$ dimensional one, it must be some special sector of the QFT which is mapped into the lower dimensional theory. We can ask what the structure of this sector is, and how the mapping to 2+1 dimensions works precisely.

These considerations motivate the study of quantum information theory in the context of holography and quantum gravity. Information theorists have already considered similar questions to the above. Indeed, a primary goal of quantum information theory is to understand how one Hilbert space may be encoded into another. This area of study has been directed more specifically at storing information in a way which is safe against errors or losses, and goes by the name of quantum error correction.

For later use, we give a very brief introduction to quantum error correcting codes here. A detailed exposition can be found in Neilson 2002 \cite{nielsen2002quantum}. Construction of an error correcting code begins with the notion of a logical Hilbert space, ${H}_{logical}$ and a physical Hilbert space $\mathcal{H}_{phys}$. The error correcting code is defined by an encoding map ${E}:{H}_{logical}\rightarrow {H}_{phys}$ and a decoding map ${D}:{H}_{phys}\rightarrow {H}_{logical}$. The dimension of $H_{phys}$ will be larger than ${H}_{logical}$, and the image of the encoding map will form a subspace ${H}_{code}$. In standard terminology, it is the subspace ${H}_{code}$ which is known as the error correcting code.

We can say that an error correcting code protects against an error $A$ if we have that, for $\rho\in H_{logical}$,
\begin{align}
(D\circ A \circ E)(\rho) = \rho.
\end{align}
There are known necessary and sufficient conditions known as the \emph{error correction conditions} for when there exists a decoding map $D$ which corrects a certain error $A$, given the subspace ${H}_{code}$. 

A particular class of error correcting codes which will be of interest are the stabilizer codes. A stabilizer code is defined using a stabilizer group, which is itself a subgroup of the Pauli group on $n$ qubits. The Pauli group $G_n$ consists of all n-fold tensor products of Pauli operators, for instance $G_2$ is generated by the elements
\begin{align}
G_2 = \langle I\otimes X, X\otimes I, I\otimes Y,Y\otimes I\rangle.
\end{align}
A stabilizer group is typically defined by a set of generators drawn from the Pauli group. The stabilizer code is defined as the joint $+1$ eigenspace of the generators of the stabilizer group. 

We can illustrate the construction of a stabilizer code using a simple example. Our physical Hilbert space will consist of three qutrits, $H_{phys} = H_1\otimes H_2\otimes H_3$ and the stabilizer group $S$ is generated by the elements\footnote{As is convention in discussions of stabilizer codes, we have omitted the $\otimes$ between the Pauli operators. Thus $ZZZ$ should be read as $Z \otimes Z \otimes Z$.}
\begin{align}\label{eq:qutritstabilizers}
S_1 &= XXX \nonumber \\
S_2 &= ZZZ 
\end{align}
with $X$ and $Z$ the generalized Pauli operators acting on qutrits\footnote{These have the property that $X\ket{n}=\ket{n-1}$ and $Z\ket{n} = e^{2\pi i n/3}\ket{n}$}. We look for the space of states which have the property $S_i\ket{\psi}=\ket{\psi}$ for all $i$. One can confirm that the following states all share this property:
\begin{align}\label{eq:codesubspace}
\ket{0_L} &= \frac{1}{\sqrt{3}}(\ket{000}+\ket{111}+\ket{222}) \nonumber \\
\ket{1_L} &= \frac{1}{\sqrt{3}}(\ket{012}+\ket{201}+\ket{120}) \nonumber \\
\ket{2_L} &= \frac{1}{\sqrt{3}}(\ket{021}+\ket{102}+\ket{210}).
\end{align}
Since each of the $S_i\ket{\psi} = \ket{\psi}$ divide the Hilbert space into thirds, there should indeed be exactly $3$ linearly independent states in the code subspace.

The subspace defined by the span of the three states $\ket{0_L},\ket{1_L},\ket{2_L}$ forms an error correcting code which can correct against the erasure of any one of the qutrits. To see this, it is sufficient to show that there is a unitary which acts on only two of the qubits which will output the logical qutrit. That is that there exists $U_{12}$ such that
\begin{align}
U_{12}\otimes I \ket{i}_{123} = \ket{i}_1 \otimes \ket{\chi}_{23}.
\end{align}
with $\ket{\chi} = \frac{1}{\sqrt{3}}(\ket{00}+\ket{11}+\ket{22})$. One can readily but tediously find the matrix elements of $U_{12}$ by substituting the states $\ket{0_L},\ket{1_L},\ket{2_L}$ into the above relation. By symmetry of the stabilizers \ref{eq:qutritstabilizers}, we know that there must also exist such a $U_{23}$ and $U_{31}$ which recover $\ket{i}$ from the $H_2\otimes H_3$ or $H_3 \otimes H_1$ Hilbert spaces.

A key lesson to be taken from the theory of quantum error correction is that one Hilbert space may be encoded into another larger one with redundancy, in the sense that loss of part of the larger Hilbert space may still allow recovery of the encoded state. This is surprising, as it seems in tension with the no-cloning theorem of quantum mechanics. It is important though that despite the recoverability of the state after loss of part of the physical Hilbert space it is only ever possible to construct one copy of the logical state by acting on the physical Hilbert space. In the example above this plays out as the need to have a majority of the qutrits, 2 of 3, to perform the reconstruction.


\section{The AdS/CFT correspondence}\label{sec:introtoadscft}

At the time of writing there is only one well established theory of quantum gravity which realizes the holographic principle. This is the AdS/CFT correspondence; it describes quantum gravity in $d+1$ dimensions in an asymptotically $AdS$ background in terms of a $d$ dimensional conformal field theory living in Minkowski space. We will give an outline of the correspondence here, and refer the reader to the literature for a more complete discussion. Some pedagogical introductions can be found in Ammon \cite{ammon2015gauge} and Polchinksi \cite{polchinski2010introduction}.

In its strongest form, the AdS/CFT correspondence asserts the equivalence of string theory in an AdS background with certain conformal field theories. We will be concerned however with a limit where the string theory is approximated by quantum fields on a classical gravity background. To be precise, the equivalence gives that
\begin{align}\label{eq:maldacena}
Z_{CFT}[J] = Z_{AdS}[\Phi\rightarrow J],
\end{align}
where the CFT partition function is computed in the presence of some external sources denoted $J$, and the AdS partition function is computed with fields subjected to boundary conditions set by those sources.   In the remainder of this section we will unpack some of the ingredients of the equation above, and outline some needed features of this correspondence. 

We will first briefly describe AdS space. In the context of tensor networks it is AdS$_{2+1}$ which is most relevant, so to simplify our discussion we will focus on that case. A metric for AdS$_{2+1}$ is
\begin{align}\label{eq:adsmetric}
ds^2 = l_{AdS}^2(-\cosh^2\rho dt^2 + d\rho^2 + \sinh^2 \rho \,d\theta^2),
\end{align}
with $\rho \in \mathbb{R}^+$, $t \in \mathbb{R}^+$, $\theta\in [0,2\pi]$, and $l_{AdS}$ a length scale known as the AdS radius. Any constant time slice has metric $ds^2 = l_{AdS}^2(d\rho^2 + \sinh^2 \rho \,d\theta^2)$ which is just the Poincare plane, $\mathbb{H}_2$. We can visualize AdS$_{2+1}$ as a cylinder, as shown in figure \ref{fig:cylinder}. 

A second useful coordinate system for describing AdS is Poincare coordinates. These do not cover the entire AdS space, but a restricted region known as the Poincare patch. The metric is
\begin{align}
ds^2 = \frac{l_{AdS}^2}{z^2}( - dt^2 + dz^2+ dx_\mu dx^\mu)
\end{align}
with $0 \leq z < \infty$, and $t,x\in \mathbb{R}$. 

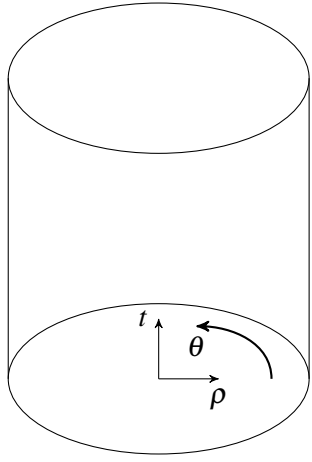
\begin{figure}
\begin{center}
\begin{tikzpicture}

\draw (0,0) ellipse (2cm and 1cm);
\draw (0,4) ellipse (2cm and 1cm);

\draw (-2,0) -- (-2,4);
\draw (2,0) -- (2,4);

\draw[->] (0,0) -- (0.8,0);
\node[below] at (0.8,0) {$\rho$};
\draw[->] (0,0) -- (0,0.8);
\node[left] at (0,0.8) {$t$};

\draw[thick,->] (1.5,0) to [out=90,in=0] (0.5,0.7);
\node[below] at (0.5,0.7) {$\theta$};

\end{tikzpicture}
\end{center}
\caption{The cylinder representing AdS$_{2+1}$. The boundary of the cylinder is at $\rho=\infty$.}
\label{fig:cylinder}
\end{figure}

AdS space has a boundary, which we can visualize as living on the surface of the cylinder. In global coordinates this boundary is at $\rho \rightarrow \infty$, while in Poincare coordinates it is at $z=0$. It is often useful to think of the boundary as the result of a limiting procedure, where a surface is drawn at $z=\epsilon$, this surface is treated as the boundary, and then later we take $\epsilon \rightarrow 0$. 

We should also very briefly introduce the other component of the AdS/CFT correspondence, conformal field theories. Most relevant for us are 1+1 dimensional conformal field theory, since these are dual according to the correspondence to string theory in AdS$_{2+1}$. A conformal field theory is a quantum field theory that has a larger symmetry group than the Lorentz group. In particular, conformal fields have no intrinsic length scale and are left invariant under scale transformations. 

We will be interested in spacelike geodesics of this AdS spacetime which are anchored on two boundary points. Starting with the metric \ref{eq:adsmetric} it is not difficult to show that these geodesics form semicircles. For later reference, we record that the length of such a geodesic is
\begin{align}
A = \frac{l_{AdS}}{2 G_N} \ln \left(\frac{L}{\epsilon} \right)
\end{align}
where $L$ is the size of the boundary interval on which the geodesic is anchored, and $z=\epsilon$ defines the cutoff surface we have used.

It is known \cite{calabrese2004entanglement} that the entanglement entropy of a single interval in a 1+1 dimensional conformal field theory is given by
\begin{align}\label{eq:ee}
S(A) = \frac{c}{3}\log \frac{L}{\epsilon}
\end{align}
where $c$ is known as the central charge, which is determined by the conformal field we are working with. $\epsilon$ is a UV cutoff indicating that modes of wavelength shorter than $\epsilon$ have been excluded from contributing to the entanglement entropy. 

Ryu and Takayanagi \cite{ryu2006holographic} pointed out that the length of the boundary anchored geodesics in AdS space and the entanglement entropy of the same boundary interval agree if one takes
\begin{align}
c = \frac{3 l_{AdS}}{2 G_N}
\end{align}
and identifies the UV cutoff in the CFT with the long distance $z$ cutoff in AdS. Indeed, both of these identifications had already been established via other lines of reasoning, which made compelling evidence for the identification of entanglement entropies in CFTs with lengths in AdS. More precisely, Ryu and Takayanagi conjectured the formula
\begin{align}
S(A) = \underset{\gamma_A}{\min} \,\,\frac{L(\gamma^A)}{4G},
\end{align}
where $\{\gamma_A\}$ are all the spacelike curves anchored on region A. This was conjectured to compute entanglement entropies in conformal fields dual to asymptotically AdS geometries. There are also corrections to the Ryu-Takayanagi formula when fields (aside from the metric) are present in the bulk. In particular the above is modified to
\begin{align}
S(A) = \underset{\gamma_A}{\min} \,\,\frac{L(\gamma^A)}{4G} + S(\rho_{W(A)})
\end{align}
where $S(\rho_{W(A)})$ is the entropy of the fields in the region enclosed by $\gamma_A$

Importantly, the Ryu-Takayanagi formula relies on their being a preferred time slicing of the AdS spacetime. The minimal surfaces $\gamma_A$ considered in the minimization must lie in a well chosen time slice. This must be the case, as its clear that if they could be chosen in any Cauchy slice of AdS with the region $A$ on its boundary one could always take the minimal length to be arbitrarily close to zero by making the slice close to lightlike. It turns out that for a spacetime with a timelike Killing vector (a static spacetime) and boundary regions which are at a constant time $t_0$, the bulk Cauchy slice defined by $t=t_0$ will contain the appropriate minimal curve.

For spacetimes without such a timelike Killing vector, which we call dynamic spacetimes, a more involved prescription for calculating boundary entanglement entropies from bulk geometry is required. In particular the Ryu-Takayanagi (RT) formula needs to be replaced by the Hubeny-Ranganmani-Takayanagi (HRT) formula, which reads
\begin{align}
S(A) = \underset{\gamma_A}{\text{extremal}} \frac{L(\gamma^A)}{4G}.
\end{align}
That is, the minimization procedure is simply replaced by an extremization. An equivalent\footnote{Actually the equivalence of HRT and maximin require the assumption of the null energy condition, but this is a very weak assumption.} formula for computing entanglement entropies using AdS geometry is the maximin formula, which states
\begin{align}
S(A) = \underset{\Sigma}{\text{max}} \left(\underset{\gamma_A}{\min}\,\, \frac{L(\gamma^A)}{4G}  \right).
\end{align}
That is, for each bulk Cauchy slice $\Sigma$ search through all possible curves $\gamma_A$, pick the shortest, and call it $\gamma_\Sigma$. The maximin formula states that the length of the longest of the $\gamma_\Sigma$ computes the entanglement entropy of $A$.

Finally, we return to the idea of AdS/CFT as an encoding of one Hilbert space into another. From the equality of partition functions \ref{eq:maldacena} we can understand how the encoding of the boundary Hilbert space into the bulk Hilbert space works. We specify a boundary state, then solve the bulk equations of motion to determine the bulk fields. To reverse the mapping, we take limits of the bulk fields as they go to the boundary and recover boundary data. 

The above gives a global prescription for determining bulk data from the boundary or vice versa. It is also interesting to ask, given a boundary region, what bulk region is determined by that boundary data? In fact, we have already learned how to determine some bulk data from a subregion of the boundary. The RT formula gives us the length of a certain bulk curve $\gamma_A$ from the data of a boundary interval $A$. Restricting attention to subintervals of $A$, we can determine the length of a infinite family of curves that sweep out the entire region enclosed by $\gamma_A$. Thus, the region enclosed by the minimal curve $\gamma_A$, which we will label $D(A)$, is a reasonable candidate for the bulk region that can be reconstructed from $A$. More covariantly, the candidate bulk region is the bulk domain of dependence of $D(A)$, which we call $W(A)$\footnote{The domain of dependence of $D(A)$ is defined as $W(A)=D(A)\cup J^+(A)\cup J^-(A)$ where $J^+(A)$ is the set of points for which all backward directed timelike curves pass though $D(A)$ and $J^-(A)$ is defined as the set of points for which all forward directed timelike curves pass through $D(A)$.}. The region $W(A)$ is known as the entanglement wedge of $A$. 

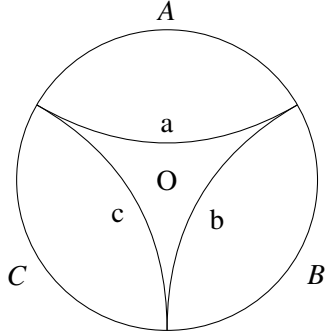
\begin{figure}
\begin{center}
\begin{tikzpicture}

\node at (0,0) {O};
\draw (0,0) circle (2);

\draw (0,-2) to [out=90,in=-150] (1.732,1);
\draw (0,-2) to [out=90,in=-30] (-1.732,1);
\draw (-1.732,1) to [out=-30,in=-150] (1.732,1);

\node[above] at (0,2) {$A$};
\node[below right] at (1.732,-1) {$B$};
\node[below left] at (-1.732,-1) {$C$};

\node[above] at (0,0.5) {a};
\node[below right] at (0.433,-0.25) {b};
\node[below left] at (-0.433,-0.25) {c};

\end{tikzpicture}
\caption{Illustration of the error correcting property of the bulk to boundary mapping. No single boundary region $A$, $B$, or $C$ is sufficient to reconstruct a bulk operator living at $O$, but any two regions $AB$, $BC$, or $CA$ is. This is because the minimal surface of $A$ (for example), which is labelled as $a$, does not enclose point $O$. However the minimal surface for $AB$, which is $c$, does. This is analogous to the three qutrit code introduced in section \ref{sec:qit}.}
\label{fig:adsthreequtrit}
\end{center}
\end{figure}

To define the idea of bulk reconstruction more concretely, we can consider an arbitrary operator $\mathcal{O}_{W(A)}$ that lives at point $x$ inside the entanglement wedge of $A$, $W(A)$. Acting with such an operator allows us to probe the bulk data, thus, our reconstruction idea implies that there should be a dual operator $\mathcal{O}_A$ that lives in the boundary theory and has the property that
\begin{align}\label{eq:errorcorrectionproperty}
\tr(\rho_{W(A)}\mathcal{O}_{W(A)}) = \tr(\rho_A \mathcal{O}_A).
\end{align}
We refer to the above as the \emph{error correction property}.

It has recently been understood that the RT formula actually implies that it is possible to reconstruct data inside the entanglement wedge \cite{harlow2017ryu}. Because a single bulk point lies in the entanglement wedge of many different boundary regions, this means the bulk to boundary encoding in AdS/CFT has the structure of an error correcting code, see \ref{fig:adsthreequtrit} for an illustration. In the context of AdS/CFT the encoding of bulk entanglement wedge data into the boundary has recently been understood \cite{jafferis2016relative,dong2016reconstruction,cotler2017entanglement}. Tensor network models were used early on to explore the entanglement wedge reconstruction idea \cite{almheiri2015bulk, pastawski2015holographic}, and consequently played a role in the development of the general understanding of the relation between the RT formula and quantum error correction. 

\chapter{Tensor networks}

In this section we develop the basics of the tensor network formalism. We emphasize two perspectives on networks which will later have connections to holography: the idea of a tensor network as a map, and the connection between the graph structure of a network and the entanglement of the state it prepares. We leave making the connection to holography to chapter \ref{sec:tnandholography}.

\section{Tensor network basics} \label{sec:basics}

\begin{figure} 
\begin{center}
\begin{tikzpicture}
\bellket{0}{0}
\node at (-2,-1) {(a)};

\bellbra{6}{0}
\node at (4,-1) {(b)};
\node at (0,3) {};
\end{tikzpicture}
\end{center}
\caption{(a) A maximally entangled state in the Hilbert space, represented in the graphical notation. (b) A maximally entangled state in the dual Hilbert space, represented in the graphical notation.}
\label{fig:bellpairs}
\end{figure}
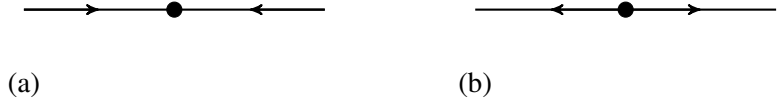

The tensor network formalism describes graphically the pattern of contraction of a set of simple objects to form a complex quantum state. The basic objects in the graphical formalism are vertices with some number of lines attached. Each vertex corresponds to a quantum state, and the lines each correspond to a ket or bra index. For instance
\begin{align}
\ket{\Psi^+} = \sum_{m=1}^D \ket{m}\ket{m}
\end{align}
is represented by figure \ref{fig:bellpairs} a. We attach a direction (inward or outward) to each line in the diagram, with inward arrows indicating ket indices and outward arrows indicating bra indices. Thus $\bra{\Psi^+}$ is represented as in figure \ref{fig:bellpairs} b. 

It will be convenient to write down quantum states without explicitly including their basis vectors. For example, the maximally entangled state $\ket{\Psi^+}$ is written $\delta_{ab}$, leaving the choice of basis implicit. In this notation ket indices are lowered and bra indices are raised, so $\bra{\Psi^+}$ becomes $\delta^{ab}$. More generally a quantum state
\begin{align} \label{eq:examplestate}
\ket{\phi} = \sum_{a,b} T_{ab} \ket{a}\ket{b}
\end{align}
is specified as $T_{ab}$. The upper and lower indices carry transformation rules with them. Since
\begin{align}
\ket{\phi} &= \sum_{a,b,c} T_a (U^\dagger)_{b}^a U^b_c \ket{c} = \sum_{b} T_b^\prime  \ket{b},
\end{align}
we see that lower indices transform according to $T_a\rightarrow T_a (U^\dagger)_{b}^a$ under a change of basis described by $\ket{b}=U^{b}_c\ket{c}$. Similarly under the same change of basis upper indices transform according to $T^a\rightarrow T^a U_a^b$.

\begin{figure}
\begin{center}
\begin{tikzpicture}[scale=0.3]
\draw[thick,postaction={on each segment={mid arrow}}] (8,3) -- (5,0);
\draw[thick,postaction={on each segment={mid arrow}}] (2,3) -- (5,0);
\draw[thick,postaction={on each segment={mid arrow}}] (8,-3) -- (5,0);
\draw[thick,postaction={on each segment={mid arrow}}] (2,-3) -- (5,0);
\draw[fill=white] (5,0) circle [radius=1]; 
\node at (5,0) {S};

\draw[thick,postaction={on each segment={mid arrow}}] (-8,3) -- (-5,0);
\draw[thick,postaction={on each segment={mid arrow}}] (-2,3) -- (-5,0);
\draw[thick,postaction={on each segment={mid arrow}}] (-8,-3) -- (-5,0);
\draw[thick,postaction={on each segment={mid arrow}}] (-2,-3) -- (-5,0);
\draw[fill=white] (-5,0) circle [radius=1];
\node at (-5,0) {T};

\draw[->,thick] (10,0)--(12,0);

\draw[fill] (21,3) circle [radius=0.25];
\draw[fill] (21,-3) circle [radius=0.25];

\draw[thick,postaction={on each segment={mid arrow}}] (21,3) to [out=0,in=135] (26,0);
\arrowedcurve{21}{3}{26}{0}{0}{135}
\arrowedcurve{21}{3}{16}{0}{180}{45}
\arrowedcurve{21}{-3}{16}{0}{180}{-45}
\arrowedcurve{21}{-3}{26}{0}{0}{-135}

\draw[thick,postaction={on each segment={mid arrow}}] (29,3) -- (26,0);
\draw[thick,postaction={on each segment={mid arrow}}] (29,-3) -- (26,0);
\draw[fill=white] (5+21,0) circle [radius=1]; 
\node at (5+21,0) {S};

\draw[thick,postaction={on each segment={mid arrow}}] (13,3) -- (16,0);
\draw[thick,postaction={on each segment={mid arrow}}] (13,-3) -- (16,0);
\draw[fill=white] (-5+21,0) circle [radius=1];
\node at (-5+21,0) {T};

\end{tikzpicture}
\end{center}
\caption{A basic example of a contraction of two quantum states into a tensor network. Algebraically, the objects at left are written $T_{abcd}$ and $S_{efgh}$. The object at right is $T_{abcd}\delta^{ce}\delta^{df}S_{efgh}$.}
\label{fig:examplenetwork}
\end{figure}
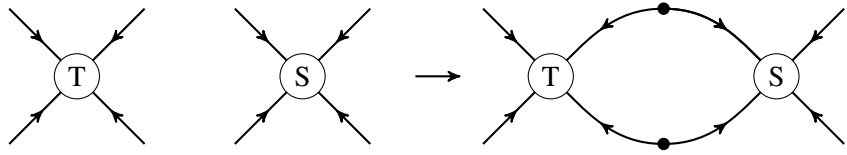

The basic operation of the tensor network formalism is the composition of two quantum states. Composition of two ket states, say $T_{abcd}$ and $S_{efgh}$ is performed by introducing maximally entangled bra states, 
\begin{align}
T_{abcd} \circ S_{efgh} \rightarrow T_{abcd}\delta^{ce}\delta^{df} S_{efgh}.
\end{align}
This is illustrated in figure \ref{fig:examplenetwork}. The contraction performed above was not the unique choice, since different pairs of indices could have been contracted. In general, to describe a pattern of contraction of two or more quantum states a graph is specified. States are associated with vertices, with each line attached to a vertex representing a particular index. The contraction is performed by placing maximally entangled pairs on the edges and connecting in going and out going lines. 

\begin{figure} 
\begin{center}
\begin{tikzpicture}[scale=0.92]

\draw[thick,fill=white] (-0.5,-0.5)--(+0.5,-0.5)--(+0.5,+0.5)--(-0.5,0.5)--(-0.5,-0.5);
\bellket{2.5}{0}
\node at (0,0) {$M$};
\draw[thick,postaction={on each segment={mid arrow}}] (-1.5,0) -- (-0.5,0);
\node at (-1.5,-1) {(a)};

\bellket{7.5}{0}
\draw[thick,fill=white] (9.5,-0.5)--(10.5,-0.5)--(10.5,+0.5)--(9.5,+0.5)--(9.5,-0.5);
\node at (10,0) {$M^T$};
\draw[thick,postaction={on each segment={mid arrow}}] (11.5,0) -- (10.5,0);
\node at (5.5,-1) {(b)};

\end{tikzpicture}
\end{center}
\caption{(a) An operator $M\otimes \mathbb{I}$ applied to a quantum state. (b) In the case where the vertex represents a maximally entangled state, the operator $M$ can be moved to the other subspace by taking the transpose.}
\label{fig:operatordiagram}
\centering
\end{figure}
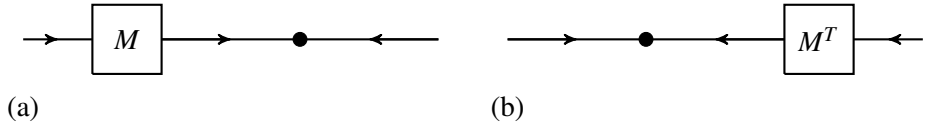

Operators acting on quantum states we represent as vertices having both inward directed and outward directed lines attached, and are written as tensors with upper and lower indices, for example ${M_a}^b$. Diagrammatically, applying an operator to a state is given by connecting lines. The algebraic equivalent is performing the appropriate sum. Thus the operator ${M_{a}}^b$ applied to a state $T_{ab}$ is represented by figure \ref{fig:operatordiagram} a or by ${M_a}^c T_{cb}$. In the case of maximally entangled states, it is straightforward to show the identity
\begin{align}
M \otimes \mathbb{I} \ket{\Psi^+} = \mathbb{I} \otimes M^T \ket{\Psi^+}.
\end{align}
We will refer to this as the \emph{transpose rule} below.

Indices can be raised and lowered by contracting with maximally entangled pairs. In particular an operator ${M_{a}}^b$ can be mapped to a state by ${M_a}^b \rightarrow M_{ab}={M_{a}}^c \delta_{cb}$, and states to operators by $M_{ab} \rightarrow {M_a}^b = M_{ac}\delta^{cb}$. In the simplest case of a state with two indices this is also known as the Choi-Jamio{\l}kowski mapping \cite{jiang2013channel} between pure bipartite states and operators. More generally, we can raise and lower indices on objects with arbitrary numbers of indices by appropriate contractions with maximally entangled pairs. 

\begin{figure} 
\centering
\begin{tikzpicture}[scale=0.3]

\node at (0,0) {(a)};
\node at (14,0) {(b)};

\arrowedsegment{1}{15}{4}{12}
\arrowedsegment{7}{15}{4}{12}
\arrowedsegment{3.29}{5.29}{1}{3}
\arrowedsegment{4.71}{5.29}{7}{3}

\arrowedsegment{15}{15}{18}{12}
\arrowedsegment{17.29}{5.29}{15}{3}
\draw[thick,->] (18,12) to [out=45,in=180] (20,14);
\arrowedcurve{20}{14}{20}{4}{0}{0}
\draw[thick,->] (20,4) -- (19.9,4);
\draw[thick,->] (20,4) to [out=180,in=-45] (18,6);

\draw[fill=white] (4,6) circle [radius=1];
\draw[fill=white] (4,12) circle [radius=1];
\node at (4,6) {$T^*$};
\node at (4,12) {$T$};

\draw[fill=white] (18,6) circle [radius=1];
\draw[fill=white] (18,12) circle [radius=1];
\node at (18,6) {$T^*$};
\node at (18,12) {$T$};

\end{tikzpicture}
\caption{(a) Representation of the density matrix $\rho_{AB}$ shown in eq. \ref{eq:densitymatrixexample} in the graphical notation. (b) Representation of the reduced density matrix $\rho_A$ in the graphical notation.}
\label{fig:densitymatrix}
\centering
\end{figure}
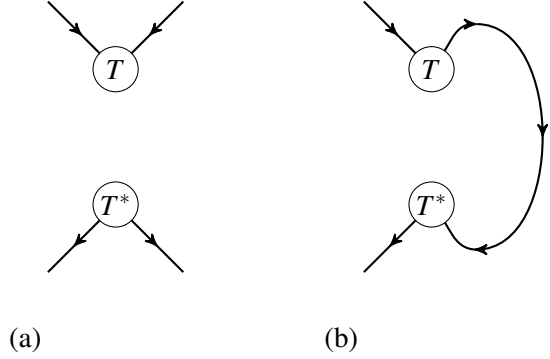

It is also possible to represent density matrices in a tensor network diagram. The density matrix corresponding to the state in eq. \ref{eq:examplestate}, given by 
\begin{align} \label{eq:densitymatrixexample}
\rho_{AB}=\sum_{ab}T_{ab} (T^*)^{cd} \ketbra{a}{c}_A\otimes\ketbra{b}{d}_B,
\end{align}
is drawn as the network shown in figure \ref{fig:densitymatrix} a. Contracting corresponding inward or outward indices in the diagram performs the partial trace. We show the diagram for $\rho_A$ in figure \ref{fig:densitymatrix} b.

In general the contraction of two properly normalized quantum states results in an unnormalized output, meaning it is necessary to add a final normalization factor after all the contractions have been performed. For this reason we frequently drop any normalization factors on our initial states, for instance writing $\ket{\Psi^+} = \sum \ket{m}\ket{m}$, since this has no effect on the final state after contraction and adding proper normalization. 

We recall an important bound on the von Neumann entropy of a subsystem of a tensor network state. Suppose we have a quantum state which is written
\begin{align}
\ket{\psi} &= \sum T_{i_1...i_n j_1...j_n}\ket{i_1}_{\bar{A}}...\ket{i_n}_{\bar{A}}\ket{j_1}_A... \ket{j_n}_A \nonumber \\
&= \sum_{IJ}T_{IJ} \ket{I}_{\bar{A}} \ket{J}_A,
\end{align}
where the capital indices stand in for a set of lower case indices, and we are interested in the entropy of the $A$ subsystem. If this state is described by a tensor network we can consider a cut $\gamma$ passing through the network and separating off the region $A$. Such a cut is specified by a path in the dual graph, which passes through a sequence of maximally entangled pairs. For each cut, there is a corresponding decomposition of the $T_{IJ}$ given by
\begin{align}
\ket{\psi} &= \sum_{IJKL}A_{IJ}\delta^{JK}{B}_{KL} \ket{I}_{\bar{A}} \ket{L}_A.
\end{align}
Now define states by $\sum_J B_{KL} \ket{L}_A = \ket{\hat{K}}_A$ and $\sum_I A_{IK} \ket{I} = \ket{\hat{K}}_{\bar{A}}$. This gives
\begin{align}
\ket{\psi} &= \sum_{K=1}^{|K|} \ket{\hat{K}}_{\bar{A}} \ket{\hat{K}}_A.
\end{align}
From this we have that $\rank(\rho_A) \leq |K|$. Since the von Neumann entropy is bounded above by the log of the rank, we have
\begin{align} \label{eq:rankbound}
S(\rho_A) \leq \log{\dim \gamma},
\end{align}
where we define the dimension of the cut by $\dim \gamma \equiv |K|$. Equality occurs when $\ket{\hat{K}}_A$ and $\ket{\hat{K}}_{\bar{A}}$ are orthonormal bases. 

\section{Maps defined from tensor networks} \label{sec:stateoncut}

Any cut which partitions the network defines two tensors, call them $C$ and $\bar{C}$, which contract to give the boundary state. That is we can write
\begin{align}
\ket{\psi}_{A\bar{A}} = \sum_{IJK} C_{IJ}\delta^{JK}\bar{C}_{KL}\ket{I_A}\ket{L_{\bar{A}}}.
\end{align}
This state can be formed by acting with the operators\footnote{If the boundary state is thought of as formed by contracting two states $\ket{C}=\sum C_{IJ}\ket{I}\ket{J}$ and $\ket{\bar{C}}=\sum \bar{C}_{KJ}\ket{K}\ket{J}$, these are just the operators that $\ket{C}$ and $\ket{\bar{C}}$ are brought to under the Choi-Jamio{\l}kowski mapping.}
\begin{align}
C = \sum_{IJ} C_{Ij_1...j_n} \ketbra{I_A}{j_{{B}_1}} ...\bra{j_{{B}_n}}, \nonumber \\
\bar{C} = \sum_{KJ} \bar{C}_{K j_1...j_n} \ketbra{K_{\bar{A}}}{j_{\bar{B}_1}} ...\bra{j_{\bar{B}_n}},
\end{align}
on a collection of maximally entangled pairs. That is
\begin{align} \label{eq:slickpsi}
\ket{\psi}_{A\bar{A}} = (C\otimes \bar{C}) \bigotimes_{i=1}^n \ket{\Psi^+}_{\bar{B}_i \bar{B}_i}.
\end{align}

In this picture $C$ and $\bar{C}$ act as maps from an interior Hilbert space onto the boundary. There is freedom in how we choose the operators $C$ and $\bar{C}$. For example, we could form the same state by contraction with a different choice of entangled states $\ket{\Psi_i}$ by writing
\begin{align} \label{eq:slickpsifree}
\ket{\psi}_{A\bar{A}} &= (C \Lambda^{-1} \otimes \bar{C} \bar{\Lambda}^{-1}) \bigotimes_{i=1}^n(\lambda^i\otimes\bar{\lambda}^i)\ket{\Psi^+}_{B_i \bar{B}_i} \nonumber \\
&= (C'\otimes \bar{C}') \bigotimes_{i=1}^n \ket{\Psi_i}_{B_i \bar{B}_i},
\end{align}
where $\Lambda = \bigotimes_i \lambda^i$ and $\bar{\Lambda}=\bigotimes_i \bar{\lambda}^i$. Additionally, we can move an operator $\lambda^i$ onto the $\bar{B}_i$ Hilbert space using the transpose rule. We could also choose operators $\Lambda$ and $\bar{\Lambda}$ which are not product. In this case we can no longer write $\bigotimes_i \ket{\Psi_i}$ for the state acted on by $C$ and $\bar{C}$. 

The general expression for $\ket{\psi}$ without placing assumptions on the form of the projecting state is
\begin{align} \label{eq:psifromgamma}
\ket{\psi}_{A\bar{A}} &= (C\otimes \bar{C}) \ket{\Psi}_{B \bar{B}},
\end{align}
where we label the Hilbert space $\bigotimes_i B_i$ by $B$ and $\bigotimes_i \bar{B}_i$ by $\bar{B}$. We give the graphical description of this expression in figure \ref{fig:cutmapprocedure} a. Eq. \ref{eq:psifromgamma} expresses the boundary state as the output of two operators acting on a state localized to the cut $\gamma$. This suggests a natural mapping to the cut, 
\begin{align} \label{eq:stateoncut}
\ket{\gamma}_{B \bar{B}}=C^\dagger C \otimes \bar{C}^\dagger \bar{C} \ket{\Psi}_{B \bar{B}}.
\end{align} 
This expression for the state on a cut is given graphically as figure \ref{fig:cutmapprocedure} b. 

\begin{figure}
\begin{subfigure}[b]{.5\textwidth}
  \centering
\begin{tikzpicture}[scale=0.3]

\draw[thick,postaction={on each segment={mid arrow}}]  (-3,-6)--(-3,-3.75);
\draw[thick,postaction={on each segment={mid arrow}}]  (0,-6)--(0,-3.75);
\draw[thick,postaction={on each segment={mid arrow}}] (3,-6)-- (3,-3.75);
\draw[thick,postaction={on each segment={mid arrow}}]  (-3,0)--(-3,-2.5);
\draw[thick,postaction={on each segment={mid arrow}}]  (0,0)--(0,-2.5);
\draw[thick,postaction={on each segment={mid arrow}}] (3,0)-- (3,-2.5);
\draw[thick,fill=white] (-5,-2) -- (5,-2) -- (5,-4) -- (-5,-4) -- (-5,-2);
\node at (0,-3) {$\ket{\Psi}$};

\foreach \x in {0,...,3}
	\draw[thick,postaction={on each segment={mid arrow}}] (30+40*\x:7) -- (30+40*\x:5);

\begin{scope}[shift={(0,-6)}]
\foreach \x in {0,...,3}
	\draw[thick,postaction={on each segment={mid arrow}}] (30+40*\x:-7) -- (30+40*\x:-5);
\end{scope} 

\draw[thick,fill=white] (-5,0) to [out=90,in=180] (0,5) to [out=0,in=90] (5,0) -- (-5,0);
\node at (0,2.25) {\Huge{$\bar{C}$}};

\draw[thick,fill=white] (-5,-6) to [out=-90,in=180] (0,-11) to [out=0,in=-90] (5,-6) -- (-5,-6);
\node at (0,-8.25) {\Huge{$C$}};

\node at (0,-20) { };

\end{tikzpicture}
\caption{}
\end{subfigure}
\begin{subfigure}[b]{.5\textwidth}
  \centering
\begin{tikzpicture}[scale=0.3]

\draw[thick,->] (0,12) to [out=-30,in=90] (5,6);
\draw[thick] (5,6) to [out=-90,in=30] (0,0);
\draw[thick,->] (0,12) to [out=-70,in=90] (2,6);
\draw[thick] (2,6) to [out=-90,in=70] (0,0);
\draw[thick,->] (0,12) to [out=-110,in=90] (-2,6);
\draw[thick] (-2,6) to [out=-90,in=110] (0,0);
\draw[thick,->] (0,12) to [out=-150,in=90] (-5,6);
\draw[thick] (-5,6) to [out=-90,in=150] (0,0); 

\draw[thick,->] (0,-18) to [out=30,in=-90] (5,-12);
\draw[thick] (5,-12) to [out=90,in=-30] (0,-6);
\draw[thick,->] (0,-18) to [out=70,in=-90] (2,-12);
\draw[thick] (2,-12) to [out=90,in=-70] (0,-6);
\draw[thick,->] (0,-18) to [out=110,in=-90] (-2,-12);
\draw[thick] (-2,-12) to [out=90,in=-110] (0,-6);
\draw[thick,->] (0,-18) to [out=150,in=-90] (-5,-12);
\draw[thick] (-5,-12) to [out=90,in=-150] (0,-6);

\draw[thick,postaction={on each segment={mid arrow}}] (0,14) -- (0,12);
\draw[thick,postaction={on each segment={mid arrow}}] (3,14) -- (3,12);
\draw[thick,postaction={on each segment={mid arrow}}] (-3,14) -- (-3,12);

\draw[thick,postaction={on each segment={mid arrow}}] (0,-20) -- (0,-18);
\draw[thick,postaction={on each segment={mid arrow}}] (3,-20) -- (3,-18);
\draw[thick,postaction={on each segment={mid arrow}}] (-3,-20) -- (-3,-18);

\draw[thick,postaction={on each segment={mid arrow}}]  (-3,-6)--(-3,-3.75);
\draw[thick,postaction={on each segment={mid arrow}}]  (0,-6)--(0,-3.75);
\draw[thick,postaction={on each segment={mid arrow}}] (3,-6)-- (3,-3.75);
\draw[thick,postaction={on each segment={mid arrow}}]  (-3,0)--(-3,-2.5);
\draw[thick,postaction={on each segment={mid arrow}}]  (0,0)--(0,-2.5);
\draw[thick,postaction={on each segment={mid arrow}}] (3,0)-- (3,-2.5);
\draw[thick,fill=white] (-5,-2) -- (5,-2) -- (5,-4) -- (-5,-4) -- (-5,-2);
\node at (0,-3) {$\ket{\Psi}$};

\draw[thick,fill=white] (-5,0) to [out=90,in=180] (0,5) to [out=0,in=90] (5,0) -- (-5,0);
\node at (0,2.25) {\Huge{$\bar{C}$}};

\draw[thick,fill=white] (-5,-6) to [out=-90,in=180] (0,-11) to [out=0,in=-90] (5,-6) -- (-5,-6);
\node at (0,-8.25) {\Huge{$C$}};

\draw[thick,fill=white] (-5,12) to [out=-90,in=180] (0,7) to [out=0,in=-90] (5,12) -- (-5,12);
\node at (0,10) {\Huge{$\bar{C}^\dagger$}};

\draw[thick,fill=white] (-5,-18) to [out=90,in=180] (0,-13) to [out=0,in=90] (5,-18) -- (-5,-18);
\node at (0,-15.75) {\Huge{$C^\dagger$}};
\end{tikzpicture}
\caption{}
\end{subfigure}
\caption{(a) Graphical description of \ref{eq:psifromgamma}, which gives a tensor network state in terms of the two block tensors $C$ and $\bar{C}$ defined by a cut $\gamma$. (b) Graphical description of eq \ref{eq:stateoncut}, which computes the state on a cut $\gamma$.}
\label{fig:cutmapprocedure}
\end{figure}
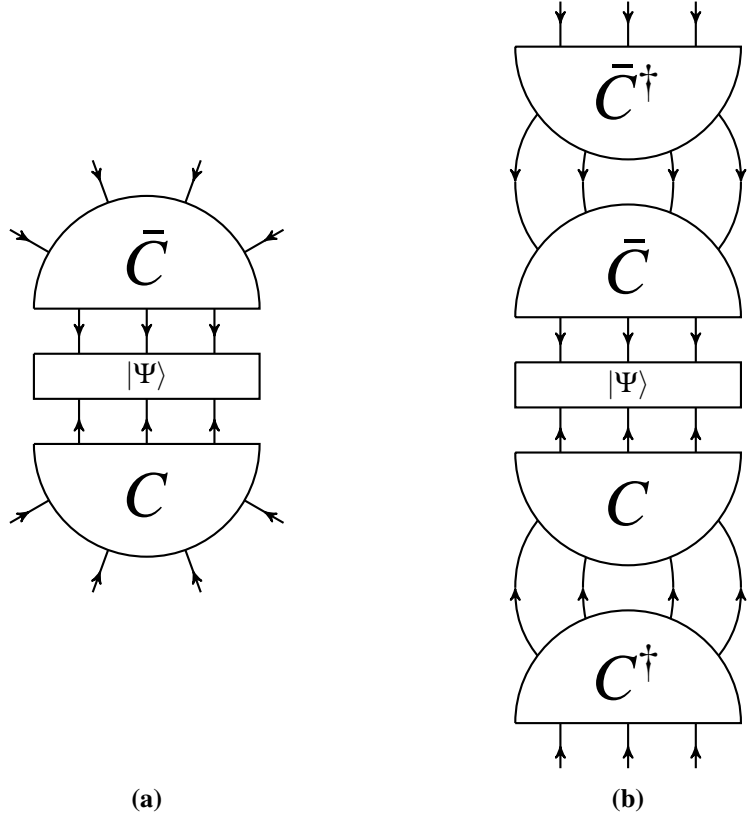   

In addition to thinking of tensor networks as maps in the sense described above, earlier literature \cite{qi2013exact, pastawski2015holographic, hayden2016holographic, yang2016bidirectional} also considers tensor networks as maps from a set of ``bulk'' uncontracted legs to the boundary legs. To build a network of this type, we place a tensor with $n+1$ legs on a vertex with $n$ edges. Contraction is performed according to the pattern of the graph as before, but now there is an extra leg associated with each vertex that remains uncontracted. It is these extra uncontracted legs that we refer to as bulk legs. A cut $\gamma$ through the network now defines a map from the cut legs plus the bulk legs to the boundary. In the presence of bulk legs we can think of the full tensor network as defining a state on the bulk and boundary legs, or as defining a mapping from bulk to boundary legs. We will have use of both perspectives.

\section{MERA and entanglement renormalization}

Outside of the holographic context the most prominent applications of tensor networks are in condensed matter theory. Tensor networks have proven central to the development of efficient numerical approximation of ground states as well as to real space renormalization techniques. These applications highlight some key features of tensor networks, which has provided and may continue to provide insight into how tensor networks may be usefully applied in holography. For this reason we outline some of these applications here. 

An early tensor network application relates to the density matrix renormalization group (DMRG)\footnote{Although similar ideas can be used to do a real space renormalization procedure in 1D, the ``DMRG'' algorithm described here is not a renormalization procedure and the naming is unfortunate.} \cite{schollwock2011density}. This is a numerical technique for determining the ground state of 1D systems. Very roughly, the procedure is as follows. We are given a Hamiltonian $H$ which acts on a lattice with some finite number of sites. After generating some initial ansatz $\ket{\Psi}$ for the ground state, we iteratively improve the accuracy of this ansatz by a repeated cutting and varying procedure. We first imagine splitting $\ket{\Psi}$ into two Hilbert spaces and writing it in the Schmidt decomposition, 
\begin{align}
\ket{\Psi} = \Sigma_{i}\Psi_i \ket{\psi_i}_A\ket{\psi_i}_B.
\end{align}
We then optimize the choice of $\Psi_i$ by minimizing $\bra{\Psi}H\ket{\Psi}$. Next, we move the cut that divides $A$ and $B$ one step over, and repeat. We continue this until reaching one end of the lattice, then turn around and sweep through the sites in the other direction. This procedure continues until the state $\ket{\Psi}$ is left unchanged by the sweeping procedure. 

The DMRG proved highly successful at describing ground states of gapped Hamiltonians. Central to this success is the DMRGs tracking of entanglement. At each step, the algorithm varies the Schmidt coefficients between two subsystems $A$ and $B$ in order to best approximate the ground state. It was later realized that DMRG could be usefully understood in terms of a tensor network representation termed the matrix product state (MPS), shown in figure \ref{fig:mpsstate}. The MPS contains links between nearest neighbour sites, accounting for its suitability in describing states with nearest neighbour interactions.

\begin{figure}
\begin{center}
\begin{tikzpicture}[scale=1]

\draw[thick,fill=black] (0,0) circle [radius=0.1];
\draw[thick,fill=black] (1,0) circle [radius=0.1];
\draw[thick,fill=black] (2,0) circle [radius=0.1];
\draw[thick,fill=black] (3,0) circle [radius=0.1];
\draw[thick,fill=black] (4,0) circle [radius=0.1];
\draw[thick,fill=black] (5,0) circle [radius=0.1];
\draw[thick,fill=black] (6,0) circle [radius=0.1];
\draw[thick,fill=black] (7,0) circle [radius=0.1];
\draw[thick,fill=black] (8,0) circle [radius=0.1];
\draw[thick,fill=black] (9,0) circle [radius=0.1];

\draw[thick] (0,0) -- (9,0);

\draw[thick] (0,0) -- (0,-1);
\draw[thick] (1,0) -- (1,-1);
\draw[thick] (2,0) -- (2,-1);
\draw[thick] (3,0) -- (3,-1);
\draw[thick] (4,0) -- (4,-1);
\draw[thick] (5,0) -- (5,-1);
\draw[thick] (6,0) -- (6,-1);
\draw[thick] (7,0) -- (7,-1);
\draw[thick] (8,0) -- (8,-1);
\draw[thick] (9,0) -- (9,-1);

\end{tikzpicture}
\end{center}
\caption{Illustration of the tensor network which prepares a matrix product state. Matrix product states are the output of the DMRG algorithm.}
\label{fig:mpsstate}
\end{figure}
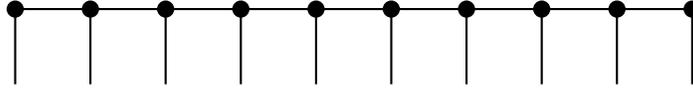

At a critical point correlation lengths diverge and we expect long range entanglement to be present in ground states. To describe a critical system efficiently we should expect to need a tensor network with a graph structure with links between sites at all distances. The MERA network \cite{vidal2008class,vidal2007entanglement} \footnote{MERA stands for Multiscale Entanglement Renormalization Ansatz, which is again sometimes a misnomer due to the wide range of uses and perspectives taken with regard to MERA.}  provides just such a representation. We illustrate the network in figure \ref{fig:MERA}.

\begin{figure}
\begin{center}
\includegraphics[scale=1]{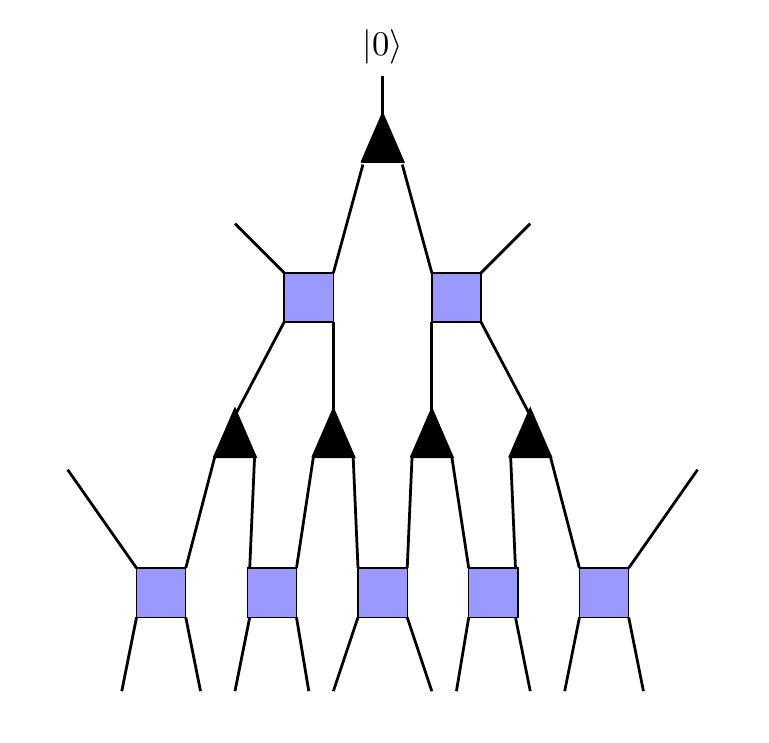}
\end{center}
\caption{Illustration of the MERA network, thought of as a quantum circuit which prepares a state. The blue squares represent unitary matrices and black triangles represent isometries. Free legs attached to blue squares are wrapped to the opposite side of the figure.}
\label{fig:MERA}
\end{figure}

At the uppermost layer of figure \ref{fig:MERA} a qubit is fed into a three index tensor, which is required to be an isometry when thought of as a map from its upper index to its two lower indices. At the next layer the outputs from the previous layer are fed into two unitary gates. The subsequent layer is then fed into a layer of isometries, and so on. The purpose of the three legged tensors is to increase the total number of legs, eventually building a large system. These legs also add entanglement, but some nearby legs will be more entangled than others if a network of only these tensors are used\footnote{The reader may draw such a network and convince themselves that some neighbouring legs are connected at the first level of the network, while others are connected many layers up.} . The four index tensors add local entanglement and ensure a translationally invariant state is prepared. In a typical application of MERA the choice of unitary and isometry is optimized to minimize the energy of the prepared state.

The MERA network may also be thought of as acting on some already known state on $N$ sites, with each layer of the MERA reducing the number of sites to $N/2$. From this perspective MERA acts as a real space renormalization procedure. The state output after one layer of MERA is a coarse grained version of the earlier layer, with the unitary-isometry structure chosen so that entanglement at the decimated length scale is removed while larger scale entanglement is preserved. The Hamiltonian also flows under this transformation, and critical points are those where the Hamiltonian is invariant under this transformation. A useful review of MERA is provided in Vidal \cite{vidal2009entanglement}.

\chapter{Tensor networks and holography}\label{sec:tnandholography}

\section{Entanglement and geometry}

The Ryu-Takayanagi formula, aside from being a powerful computational tool, reveals a startling connection between entanglement in holographic CFTs and their bulk dual gravity theories. This was emphasized early on by Van Raamsdonk \cite{van2010building} who considered the thermofield double state,
\begin{align}
\ket{\Psi} = \sum_i e^{-\beta E_i/2} \ket{E_i}_A\otimes \ket{E_i}_B,
\end{align}
where $A$ and $B$ are the Hilbert spaces for two CFTs. Van Raamsdonk recalled that this state had been understood to correspond to a wormhole geometry in the bulk \cite{maldacena2003eternal}. This is already surprising as the A and B CFTs are non-interacting and their only relation is that they have been put in this entangled state. It seems that the entanglement between the A and B subsystems is somehow responsible for the bulk wormhole connection between the spacetimes.

\begin{figure}
\begin{center}
\includegraphics[scale=1]{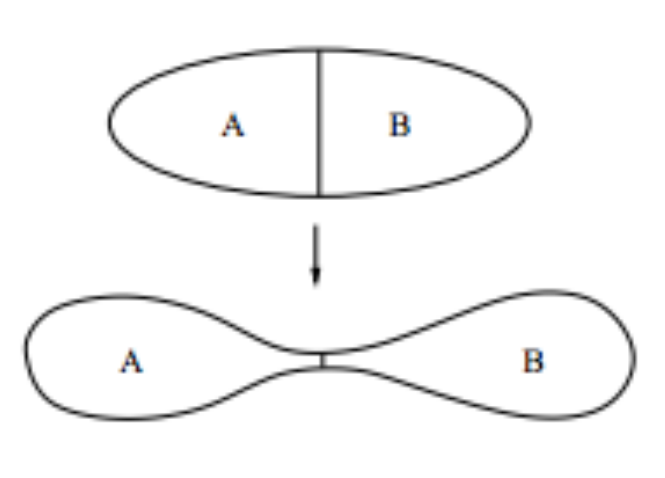}
\end{center}
\caption{As the entanglement between the A and B regions in the CFT is removed, the bulk geometry pinches and then is pulled apart.}
\end{figure}

To get a more quantitative handle on this one can consider tuning the parameter $\beta$ to decrease the entanglement between $A$ and $B$. We can measure the entanglement using the Von Neumann entropy of $A$, and relate this to the area of the wormhole neck via the RT formula. When this is done it is found that as the entanglement decreases the wormhole neck narrows before finally pinching off as the entanglement vanishes, leaving two disconnected spacetimes. Following on this perspective of \emph{entanglement builds geometry} various authors have pursued a line of work which takes as starting point the RT formula and tries to extract gravitational physics from properties of entanglement \cite{lashkari2014gravitational,swingle2014universality}. This program has successfully recovered Einsteins equations to first and second order \cite{faulkner2014gravitation,faulkner2017nonlinear}, and proven several positive energy theorems \cite{lashkari2016gravitational}.  

Tensor networks realize the entanglement-geometry connection in an immediate way. Consider for example a quantum state which contains two unentangled subsystems. It will be possible to prepare such a quantum state using two disconnected networks, and interpreting the networks graph as a discretized geometry we immediately find that the spacetime dual to the state contains two disconnected regions, just as in the AdS/CFT example. More quantitatively, the entropy bound \ref{eq:rankbound} which reads 
\begin{align}\label{eq:repeatbound}
S(A) \leq \log \dim \gamma
\end{align}
for tensor networks tells us that given some amount of entanglement between two subsystems any cut through the network which divides the two regions must contain some minimal number of legs. 

It is interesting that in the tensor network picture entanglement only specifies a minimal number of legs required. This implies the graph geometry is not totally specified by the entanglement properties of the quantum state it prepares, rather, there are many networks with potentially distinct graph structures which prepare the same state. This has lead to the introduction of various additional requirements on the tensor network in cases where one would like to interpret it geometrically. Typically some requirement is added which ensures the bound \ref{eq:repeatbound} is saturated, allowing entanglement entropies to be determined from the graph structure alone. In section \ref{sec:happynetworksreview} and section \ref{sec:randomreview} we will see two examples of this.

\section{Early efforts and constraints on holographic tensor networks}

Tensor network states display a connection between their graph geometry and entanglement properties, and so it is natural to wonder if tensor networks can be used to build toy models of AdS/CFT, or even to define more general notions of holography than AdS/CFT. This idea was first pointed out and pursued by Swingle \cite{swingle2012constructing} and has been an area of active interest since. 

Early efforts at establishing contact between tensor networks and AdS/CFT focused on the MERA network. This is natural as MERA had already proven its usefulness in approximating states in a CFT. Further, Swingle noted the minimal surfaces in MERA are similar to the AdS minimal surfaces, and bound \ref{eq:rankbound} on the entropy then connects these minimal surfaces to boundary entanglement. Qi developed this proposal by including a set of bulk legs in the MERA network and considering the tensor network as a map between bulk and boundary Hilbert spaces \cite{qi2013exact}. 

Another perspective taken early on was to look for general conditions constraining all holographic tensor networks, rather than to construct a specific model. Such arguments lead Bao et al. \cite{bao2015consistency} to the realization that MERA could only hope to describe AdS geometry at lengths greater than the AdS radius. It was also found that there was no dimension for the bulk Hilbert space which would satisfy both the RT formula and the Bousso bound\footnote{The Bousso bound is a bound on the entropy of a certain light-sheet which generalizes the $S \leq A/4G$ bound coming from avoiding black hole formation.} for a MERA network. 

There were two developments which side stepped these limitations. First, some authors pursued other choices of tensor network, beginning with the networks built from perfect tensors \cite{pastawski2015holographic}. These networks have the advantage of realizing the AdS geometry in a more direct way - minimal cuts in perfect networks realize the RT formula exactly. These also have the advantage of being translationally symmetric and isotropic, whereas a MERA network has a directionality. However, they have the significant disadvantage that the boundary state they prepare does not approximate a CFT state. Nonetheless these networks realize both the RT formula and the error correction properties of AdS/CFT and have consequently proven to be interesting toy models. We discuss them at some length beginning in section \ref{sec:happynetworksreview}.

The second direction which side steps the constraints of Bao et al. \cite{bao2015consistency} is to reinterpret the MERA network geometry not as the spatial geometry of AdS, but rather as de Sitter. This proposal first appeared in Beny \cite{beny2013causal}, where the causal structure of the MERA was identified as a discrete version of that appearing in de Sitter space. 

A second development by Czech et al. \cite{czech2015integral, czech2016tensor} also associated MERA with de Sitter space, but beginning with a very different starting point. There, the authors considered the mathematical space known as \emph{kinematic space}. For any space with a measure, we can define the associated kinematic space by parameterizing the set of all geodesics and can equip this new space with a natural measure. Further, for the case of AdS$_3$ it was possible to determine a metric for the new space and identify it as dS$_2$. This new kinematic space turns out to encode entanglement properties of the CFT is a direct way. In particular, the volume of a region in kinematic space is computed as the conditional mutual information of a set of three boundary regions. The interpretation of the MERA network as geometric is much more natural if one identifies the MERA as approximating kinematic space. In fact, conditional mutual information in MERA is computed approximately by counting vertices in a certain region, similar to the volume calculation in kinematic space. Further, the causal structure identified earlier in the MERA network can be matched to the causal structure of kinematic space. 

We will focus on networks which approximate spatial slices of AdS in this thesis. However, our lack of discussion of MERA and kinematic space should not be taken as representative of the importance of this work. Both the real space and kinematic space networks are interesting. Typically, real space models do not effectively approximate CFT ground states but can realize the error correction property, while kinematic space models do the reverse. A possible exception to this is the recent hyper-invariant models \cite{evenbly2017hyper}.

\section{The subregion isometry property}\label{sec:subregionisometry}

There are various aspect of AdS/CFT that we might want to capture in a toy model, but most tensor network literature has focused on two: the entanglement-geometry connection, as made precise in the Ryu-Takayanagi formula, and the error correction property. In fact, we can show fairly easily that both of these properties can be realized in a tensor network given that the tensor network has what we will call the subregion isometry property. 

\begin{figure}
\begin{center}
\includegraphics[scale=1]{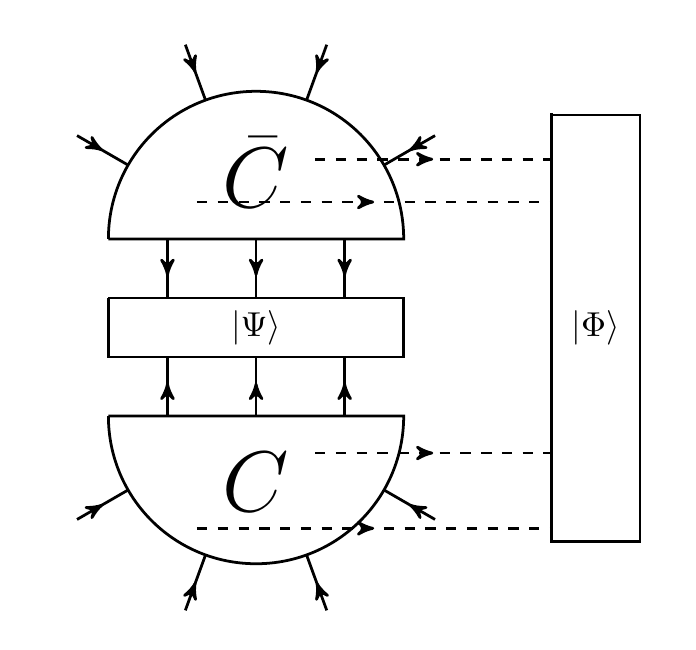}
\end{center}
\caption{A tensor network which includes bulk legs. The state on a cut $\ket{\Psi}_{B\bar{B}}$ is associated to any path through the dual graph $\gamma$.}
\label{fig:mapwithbulk}
\end{figure}

We illustrate a tensor network schematically in figure \ref{fig:mapwithbulk}. A cut $\gamma$ has been chosen which divides the network into two regions, labelled $C,\bar{C}$ and defines a state on the cut $\ket{\Psi}_{B\bar{B}}$ as discussed in section \ref{sec:stateoncut}. There are bulk legs associated with the regions $C$ and $\bar{C}$, and projected into these is the bulk state $\ket{\Phi}_{W(A)W(\bar{A})}$. Given a network with these basic components we define the subregion isometry property as follows:
\begin{definition}
A tensor network is said to have the \textbf{subregion isometry property} when, given a cut of minimal length $\gamma_A$ anchored on a boundary region $A$, the tensor network defines maps $C:H_{W(A)B}\rightarrow A$ and $\bar{C}:H_{W(\bar{A})\bar{B}}\rightarrow H_{\bar{A}}$ which are both isometries. 
\end{definition}

Much of the tensor network literature consists of methods for constructing networks which have this property. We discuss HaPPY networks and random tensor networks from this perspective in sections \ref{sec:happynetworksreview} and \ref{sec:randomreview}, respectively. For now, let us assume it is possible to construct such networks and investigate the consequences. Suppose then that the network shown schematically in figure \ref{fig:mapwithbulk} has the subregion isometry property. Then algebraically the state prepared is
\begin{align}
\ket{\psi} = C\otimes \bar{C} \left(\ket{\Psi}_{B\bar{B}}\otimes \ket{\Phi}_{W(A)W(\bar{A})}\right).
\end{align}
We will be interested in the reduced density matrix of a subregion $A$,
\begin{align}\label{eq:densitymatrixwithbulk}
\rho_A &= \tr_{\bar{A}} (C\otimes \bar{C} \left(\ketbra{\Psi}{\Psi}_{B\bar{B}}\otimes \ketbra{\Phi}{\Phi}_{W(A)W(\bar{A})}\right) C^\dagger \otimes \bar{C}^\dagger).
\end{align}
If we now use the cyclic property of the trace and that $\bar{C}^\dagger C = \mathbb{I}$ we find
\begin{align}\label{eq:simplerdensitymatrixwithbulk}
\rho_A = C \left[\tr_{\bar{B}}(\ketbra{\Psi}{\Psi}_{B\bar{B}})\otimes \tr_{W(\bar{A})}(\ketbra{\Phi}{\Phi}_{W(A)W(\bar{A})})\right] C^\dagger.
\end{align}

From here we can straightforwardly see the Ryu-Takayanagi property. Since the von Neumann entropy is unchanged under conjugation by an isometry, we have
\begin{align}
S(\rho_A)  &= S(\tr_{\bar{B}}(\ketbra{\Psi}{\Psi}_{B\bar{B}}) + S(\rho_{W(A)}) \nonumber \\
&= |\gamma_A|\log D + S(\rho_{W(A)}).
\end{align}
where $|\gamma_A|$ represents the number of legs cut by $\gamma$, and $D$ is the dimension of a single leg. Typically we identify $|\gamma_A|\log D$ as the length of the cut, so that 
\begin{align}
S(\rho_A) = \underset{\gamma_A}{\text{min}} \, L(\gamma_A) + S(\rho_{W(A)}).
\end{align}
Which is exactly the Ryu-Takayanagi formula, including the bulk entropy term.

The subregion isometry property also leads immediately to the error correction property. Recall that the error correction property can be stated precisely by the requirement \ref{eq:errorcorrectionproperty}, which requires there be a boundary operator $\mathcal{O}_A$ living on the region $A$ for every bulk operator $\mathcal{O}_{W(A)}$ living in the entanglement wedge and such that
\begin{align}
\tr (\rho_A \mathcal{O}_A) = \tr (\rho_{W(A)} \mathcal{O}_{W(A)}).
\end{align}
To construct the operator $\mathcal{O}_A$ from $\mathcal{O}_{W(A)}$ in a tensor network, we define
\begin{align}
\mathcal{O}_A \equiv C (\mathcal{O}_{W(A)}\otimes \mathbb{I}_B )C^\dagger.
\end{align}
Then it is a simple calculation to check \ref{eq:errorcorrectionproperty},
\begin{align}
\tr(\rho_A \mathcal{O}_A) = \tr( \rho_A C (\mathcal{O}_{W(A)}\otimes \mathbb{I}_B) C^\dagger ) = \tr( C^\dagger \rho_A C (\mathcal{O}_{W(A)}\otimes \mathbb{I}_B))
\end{align}
Now from \ref{eq:simplerdensitymatrixwithbulk} we have that
\begin{align}
C^\dagger \rho_A C = \rho_{W(A)}\otimes \rho_B,
\end{align}
which leads to
\begin{align}
\tr (\rho_A \mathcal{O}_A) = \tr ((\rho_W(A) \otimes \rho_B) (\mathcal{O}_{W(A)}\otimes \mathbb{I}_B) ) = \tr(\rho_{W(A)}\mathcal{O}_{W(A)}).
\end{align}
We should also note that expression \ref{eq:simplerdensitymatrixwithbulk} has already made use of $\bar{C}$ being an isometry, so both $C$ and $\bar{C}$ being isometries is used in the proof of both the error correction and RT formulas.

\chapter{Examples of real space holographic tensor networks}

\section{HaPPY networks}\label{sec:happynetworksreview}

Due to the shortcomings of the MERA in describing AdS, some authors began pursuing other classes of tensor network as possible toy models of AdS/CFT. HaPPY networks \cite{pastawski2015holographic} are one model which has been introduced; they have the property that cuts which cross a minimal number of legs saturate the bound \ref{eq:rankbound}. This gives HaPPY networks a precise connection between entanglement and graph geometry as they satisfy the RT formula exactly. We outline some facts about HaPPY networks in this section.

The basic building block of a HaPPY network is a perfect tensor. We remind the reader of the definition of a perfect tensor below. 
\begin{definition}
A \textbf{perfect tensor} is a tensor $T_{a_1 a_2...a_{2n}}$ with an even number of indices and having the property that
\begin{align} \label{eq:perfectioncondition}
T_{a_1...a_n a_{n+1}...a_{2n}} (T^*)^{b_1...b_n a_{n+1}...a_{2n}} = \delta_{a_1}^{b_1}...\delta_{a_n}^{b_n},
\end{align}
where the $a_{n+1}...a_{2n}$ can be chosen to be any of the $2n$ legs of the tensor.
\end{definition}

We can also raise and lower legs on the left side of \ref{eq:perfectioncondition}, giving
\begin{align}\label{eq:unitarycondition}
{T_{a_1...a_n}}^{a_{n+1}...a_{2n}} {(T^*)^{b_1...b_n}}_{a_{n+1}...a_{2n}} = \delta_{a_1}^{b_1}...\delta_{a_n}^{b_n}.
\end{align}
This shows we can think of the perfection condition as the statement that the tensor defines a unitary transformation from any set of $n$ legs to the complement. It follows by contracting indices on both sides of \ref{eq:unitarycondition} that perfect tensors define isometries from any subset of legs of size $k<n$ to the complement. We illustrate the perfection condition in figure \ref{fig:perfectcondition}. 

To explicitly construct a perfect tensor, we can begin by thinking about the three qutrit error correcting code discussed in section \ref{sec:qit}. Recall the error correcting code consisted of a code subspace $H_{code}$ spanned by three states $\ket{0_L},\ket{1_L}, \ket{2_L}$ which were written explicitly in equation \ref{eq:codesubspace}. To construct a perfect tensor, we take the state
\begin{align}\label{eq:perfectqutrit}
\ket{\Psi}_{b123} = \ket{0}_b\otimes \ket{0_L}_{123} + \ket{1}_b\otimes \ket{1_L}_{123} + \ket{2}_b\otimes \ket{2_L}_{123} 
\end{align}
which one can check explicitly defines a perfect tensor. The four index tensor defined by \ref{eq:perfectqutrit} actually acts as the encoding map $\ket{i}\rightarrow \ket{i_L}$. The mapping is given by
\begin{align}
_{b}\braket{i}{\Psi}_{b123} = \ket{i_L}_{123}.
\end{align}

A useful operation involving perfect tensors is operator pushing. Suppose we have a perfect tensor $T$ and an operator $\mathcal{O}$ which acts on three legs. Then we can rewrite the tensor $\mathcal{O}T$ as $T \mathcal{O}'$ by defining $\mathcal{O}'=T^\dagger \mathcal{O} T$. We illustrate this in figure \ref{fig:operatorpushing}. An operator acting on a single leg of a $2n$ leg perfect tensor can be pushed through to any $n$ legs, but in general the operator $\mathcal{O}'$ will not act as a tensor product across those legs.  

To construct a HaPPY network, perfect tensors are placed on the vertices of a graph with a non-positive curvature condition\footnote{By non-positive curvature it is meant that distance (measured in number of legs cut) between points in the dual graph has no maximum away from the boundary.}. Reference \cite{pastawski2015holographic} which introduced HaPPY networks does not keep track of the distinction between upper and lower indices in their construction, so to translate their construction to the language used here we must consider a maximally entangled state being placed along every edge of this non-positively curved graph. This done, we may perform the contraction, leaving a boundary state whose entanglement entropies saturate \ref{eq:rankbound}.

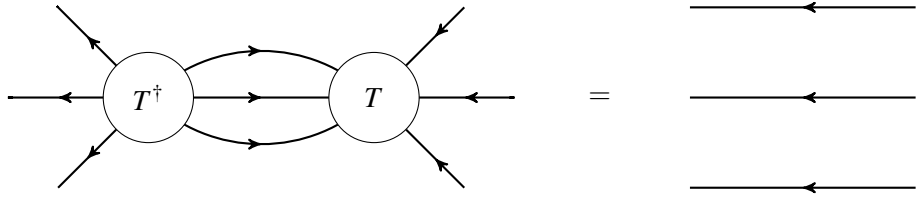
\begin{figure}
\begin{center}
\begin{tikzpicture}[scale=0.6]
\draw[thick, postaction={on each segment={mid arrow}}] (0,0) to [out=0,in=180] (5,0);
\draw[thick, postaction={on each segment={mid arrow}}] (0,0) to [out=45,in=135] (5,0);
\draw[thick, postaction={on each segment={mid arrow}}] (0,0) to [out=-45,in=-135] (5,0);
\draw[thick, postaction={on each segment={mid arrow}}] (-1,0)  to [out=0,in=180] (-3,0);
\draw[thick, postaction={on each segment={mid arrow}}] (-0.71,0.71)  to [out=-45,in=135] (-2,2);
\draw[thick, postaction={on each segment={mid arrow}}] (-0.71,-0.71) to [out=45,in=-135] (-2,-2);
\draw[thick, postaction={on each segment={mid arrow}}] (8,0) to [out=0,in=180] (6,0);
\draw[thick, postaction={on each segment={mid arrow}}] (7,2) to [out=45,in=-135] (5+0.71,0.71);
\draw[thick, postaction={on each segment={mid arrow}}] (7,-2) to [out=-45,in=135] (5+0.71,-0.71);
\draw[fill=white] (0,0) circle [radius=1];
\node at (0,0) {$T^\dagger$};
\draw[fill=white] (5,0) circle [radius=1];
\node at (5,0) {$T$};

\node at (10,0) {$=$};
\draw[thick, postaction={on each segment={mid arrow}}] (17,0)--(12,0);
\draw[thick, postaction={on each segment={mid arrow}}] (17,2)--(12,2);
\draw[thick, postaction={on each segment={mid arrow}}] (17,-2)--(12,-2);
\end{tikzpicture}
\end{center}
\caption{Illustration of the defining condition for perfect tensors. The same equality must hold when any subset consisting of half the legs is contracted.}
\label{fig:perfectcondition}
\end{figure}

In the case of HaPPY networks the length of curves through the graph is defined by
\begin{align}\label{eq:graphlength}
L_G(\gamma) = \log (\dim \gamma).
\end{align}
Herein we will refer to this as the graph length. In section \ref{sec:lengthdefined} we will discuss an alternative notion of length in the network.

\begin{figure} 
\begin{center}
\begin{tikzpicture}[scale=1]

\arrowedcurve{0.5}{0}{1.5}{0}{0}{180}
\arrowedcurve{0.5}{0.3}{1.5}{0.3}{0}{180}
\arrowedcurve{0.5}{-0.3}{1.5}{-0.3}{0}{180}
\arrowedcurve{-1.5}{0}{-0.5}{0}{0}{180}
\arrowedcurve{-1.5}{-0.3}{-0.5}{-0.3}{0}{180}
\arrowedcurve{-1.5}{0.3}{-0.5}{0.3}{0}{180}
\arrowedcurve{3.5}{0}{2.5}{0}{0}{180}
\arrowedcurve{3.5}{-0.3}{2.5}{-0.3}{0}{180}
\arrowedcurve{3.5}{0.3}{2.5}{0.3}{0}{180}

\node at (4.85,0) {$=$};

\arrowedcurve{6}{0}{7}{0}{0}{180}
\arrowedcurve{6}{0.3}{7}{0.3}{0}{180}
\arrowedcurve{6}{-0.3}{7}{-0.3}{0}{180}
\arrowedcurve{9}{0}{8}{0}{0}{180}
\arrowedcurve{9}{-0.3}{8}{-0.3}{0}{180}
\arrowedcurve{9}{0.3}{8}{0.3}{0}{180}
\arrowedcurve{11}{0}{10}{0}{0}{180}
\arrowedcurve{11}{-0.3}{10}{-0.3}{0}{180}
\arrowedcurve{11}{0.3}{10}{0.3}{0}{180}

\draw[thick,fill=white] (-0.5,-0.5)--(+0.5,-0.5)--(+0.5,+0.5)--(-0.5,+0.5)--(-0.5,-0.5);
\draw[thick,fill=white] (2-0.5,-0.5)--(2+0.5,-0.5)--(2+0.5,0.5)--(2-0.5,0.5)--(2-0.5,-0.5);
\node at (0,0) {$\mathcal{O}$};
\node at (2,0) {T};

\draw[thick,fill=white] (7,-0.5)--(8,-0.5)--(8,0.5)--(7,0.5)--(7,-0.5);
\draw[thick,fill=white] (9.5-0.5,-0.5)--(9.5+0.5,-0.5)--(9.5+0.5,+0.5)--(9.5-0.5,+0.5)--(9.5-0.5,-0.5);
\node at (9.5,0) {$\mathcal{O}'$};
\node at (7.5,0) {$T$};

\end{tikzpicture}
\end{center}
\caption{Illustration of the operator pushing operation. An operator $\mathcal{O}$ acting on a subset of size $n$ of a perfect tensor with $2n$ legs is equivalent to an operator $\mathcal{O}'=T^\dagger \mathcal{O} T$ acting on the complement.}
\label{fig:operatorpushing}
\end{figure}
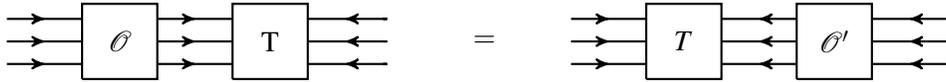

The key result regarding HaPPY networks which gives them a precise entanglement-geometry connection is
\begin{align} \label{eq:RT}
S(A) = \underset{\gamma^A}{\min} \, L_G(\gamma^A),
\end{align}
where $A$ is a single boundary interval and the minimization is taken over cuts $\gamma_A$ enclosing $A$. To show this, the authors show that both sides of a minimal cut can be interpreted as a unitary circuit from the cut legs and a subset of the boundary legs to the remainder of the boundary legs. We restate this result in a slightly changed language as follows.

\begin{theorem} \label{thm:happy}
In a HaPPY network, a cut which is anchored on a boundary interval $A$ and crosses a minimal number of legs defines a map from the cut legs to the interval $A$ which is an isometry.
\end{theorem}

We will refer to cuts of a network that define isometries on both sides as \emph{isometric cuts}. This result allows us to calculate the state defined on a cut as given in eq. \ref{eq:stateoncut} whenever the cut is minimal. Recalling that all of the contractions in a HaPPY network are performed with maximally entangled pairs, \ref{eq:stateoncut} becomes
\begin{align} 
\ket{\gamma}_{B \bar{B}} = (C^\dagger C\otimes \bar{C}^\dagger \bar{C}) \bigotimes_{i=1}^n \ket{\Psi^+}_{B_i \bar{B}_i}.
\end{align}
When the cut is minimal theorem \ref{thm:happy} gives $C^\dagger C=\mathbb{I}$ and $\bar{C}^\dagger \bar{C}=\mathbb{I}$, so the state on the cut is just a collection of maximally entangled pairs, with one pair for each edge cut by $\gamma$. 

In fact, we can straightforwardly extend theorem \ref{thm:happy}  to an if and only if statement as follows.

\begin{theorem}\label{thm:converse}
In a HaPPY network, a cut $\gamma$ enclosing a single boundary interval defines an isometry on both sides if and only if $\log \dim \gamma$ is minimal.
\end{theorem}

\begin{proof}
That a cut being minimal implies the maps it defines are isometries is given as theorem \ref{thm:happy}. 

Next we show that an isometric cut is minimal. Consider the boundary state as written in \ref{eq:psifromgamma}, which corresponds to the diagram in figure \ref{fig:cutmapprocedure}a. To form the density matrix on a region $A$ we draw an arrow reversed duplicate of \ref{fig:cutmapprocedure}a, and contract the $\bar{A}$ legs. Then since $\bar{C}^\dagger \bar{C} = \mathbb{I}$ and $\ket{\Psi} = \bigotimes_i \ket{\Psi^+}$ we are left with
\begin{align}
\rho_A = C C^\dagger.
\end{align}
To get the normalization factor note that $\tr(CC^\dagger) = \tr (C^\dagger C) = \log \dim \gamma$. Further, since $CC^\dagger$ is a projector its non-zero eigenvalues are equal to one. Using these two facts we have that
\begin{align} \label{eq:bb}
S(\rho_A) = \log \dim \gamma.
\end{align}
At the same time, the bound \ref{eq:rankbound} gives that $S(\rho_A) \leq \log\dim \gamma'$ for any cut $\gamma'$ in the network. Combining this with \ref{eq:bb} we have
\begin{align}
\log \dim \gamma \leq \log \dim \gamma'
\end{align}
for any cut in the network. Thus any cut which defines an isometry on both sides is minimal.\end{proof}

As a consequence of theorem \ref{thm:converse} the RT formula for HaPPY networks can be restated as
\begin{align} \label{eq:newRT}
S(A) = L(\gamma_{iso}^A).
\end{align}
Here $\gamma$ is any isometric cut enclosing $A$. We will refer to this as the \emph{isometric cut formula}\footnote{As with theorems \ref{thm:happy} and \ref{thm:converse} the isometric cut formula is proven only for boundary regions consisting of a single interval.}.

It would be useful to extend theorem \ref{thm:happy} to the case where the network includes bulk legs so as to establish the isometric subregion property for HaPPY networks. One HaPPY network with bulk legs which appears to have this property is shown in figure \ref{fig:treenetwork}. Each vertex has three planar legs and one bulk leg. The network is a tree in the computer science sense, and we can assign the root node a label of $0$, the vertices one edge away from the root $1$, and so on. The edges may then be assigned a directionality by having them point from lower to higher label numbers. By checking cases, the reader can convince themselves that minimal cuts define isometric mappings from bulk plus cut legs to the boundary.

\begin{figure}
\begin{center}
\includegraphics[scale=0.7]{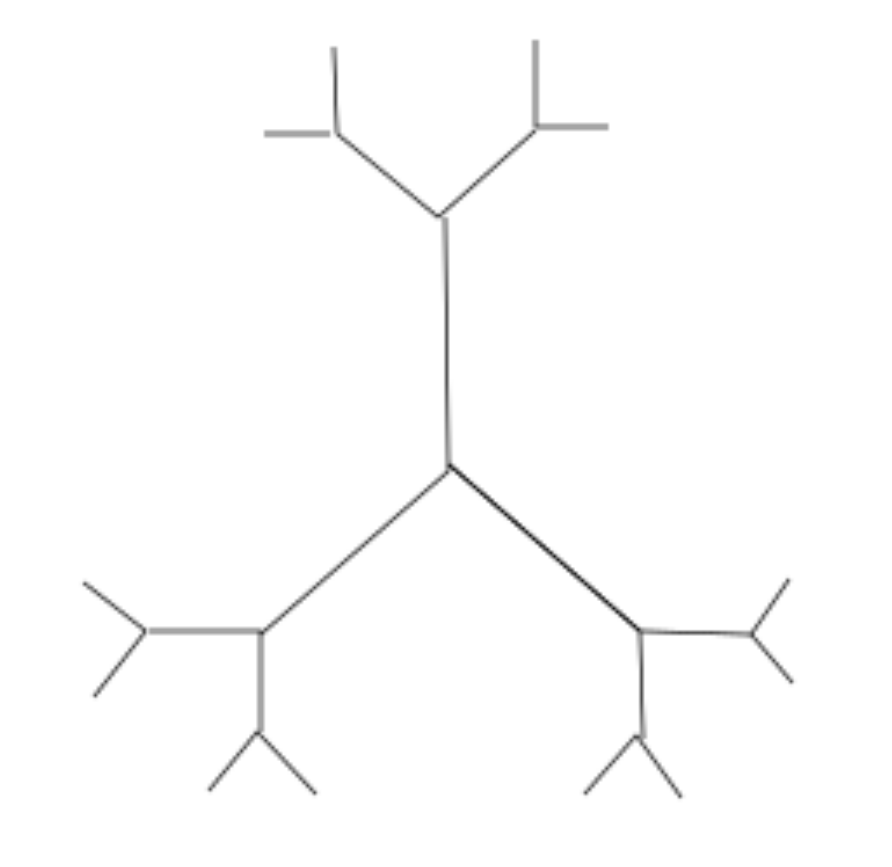}
\end{center}
\caption{A tree graph on which the 4 legged perfect tensor can be placed to build a tensor network. The resulting network seems to have the subregion isometry property, as can be argued by checking cases.}
\label{fig:treenetwork}
\end{figure}

\section{Random tensor networks}\label{sec:randomreview}

The key property of a HaPPY network which gave it the RT formula was that a minimal cut through the network defined isometries on both sides. However, HaPPY networks with bulk legs have not been proven to have the subregion isometry property, a key feature we would like in a holographic toy model. There are now other classes of network which are known to have this property however, the most general and versatile among them being the random tensor networks which we now discuss.

The reduced density matrix of any subsystem of a state defined by a perfect tensor is maximally mixed. In fact, we can characterize a perfect tensor as one with maximal entanglement across any division of the Hilbert space. This is actually a generic property: a state drawn at random according to the Harr measure will with high probability have nearly maximal entanglement across its subsystems. This expectation becomes more and more precise as the dimension $D$ of each subsystem becomes large. 

There is a tradition in quantum information theory of studying random quantum states \cite{collins2016random} and there are well developed techniques available for their study. Because of this, and because the closely related perfect tensors have already proven useful, it is natural to consider tensor networks built from random tensors in the holographic context.

To build a random tensor network one begins with a graph. We will give the vertices labels $x$, and place states $\ket{V_x}$ at each vertex. $\ket{V_x}$ is defined by a tensor with $n_x+1$ legs, where $n_x$ is the number of edges connected to vertex $x$. The states $\ket{V_x}$ are drawn independently at random according to the Harr measure by beginning with a reference state $\ket{0_x}$ and acting with a random unitary $U$.

Next, maximally entangled pairs are placed on the edges of the graph and used to contract legs. One uncontracted leg is left on each vertex which will serve as a bulk leg. We will then project some bulk state into these legs, leaving us with a network that defines a boundary state. We can write the uncontracted state as a density matrix,
\begin{align}\label{eq:randomrho}
\rho = \bigotimes_x \ketbra{V_x}{V_x}
\end{align}
and the contraction can be expressed using the partial trace,
\begin{align}
\rho = \tr\left( \rho_P \left(\bigotimes_x \ketbra{V_x}{V_x}\right) \right)
\end{align}
with 
\begin{align}
\rho_P = \rho_b \otimes\left( \bigotimes_x \ketbra{xy}{xy} \right).
\end{align}
The dimension of a leg connecting vertices $x$ and $y$ will be referred to as the \emph{bond dimension} $D_{xy}$. 

It is possible to make use of random matrix techniques to calculate entanglement properties of the state eq. \ref{eq:randomrho}. We will only outline a few of the key steps in doing this here and refer the reader to Hayden \cite{hayden2016holographic} for details. First, one can realize that it is easier to calculate the 2nd Renyi entropy than the Von Neumann entropy when using random tensors. Later it can be argued that the Renyi and Von Neumann entropies agree for these particular states. The 2nd Renyi entropy is defined by
\begin{align}\label{eq:bartwo}
e^{-S_2(\rho_A)} = \frac{\tr[(\rho\otimes \rho)\mathcal{F}_A]}{\tr[\rho \otimes \rho]}
\end{align}
where $\mathcal{F}_A$ is known as a swap operator and is defined by
\begin{align}
\mathcal{F}_A(\ket{n}_{A1}\otimes\ket{m}_{\bar{A}1}\otimes\ket{n'}_{A2}\otimes\ket{m'}_{\bar{A}2}) = \ket{n'}_{A1}\otimes\ket{m}_{\bar{A}1}\otimes\ket{n}_{A2}\otimes\ket{m'}_{\bar{A}2}. 
\end{align}
The convenience of the Renyi entropy lies in the fact that we can average over all tensors $\ket{V_x}$ at each site before doing the projection and the trace, so that
\begin{align}\label{eq:barone}
\overline{\tr[(\rho\otimes \rho)\mathcal{F}_A]} = \tr \left[(\rho_P\otimes\rho_P)\mathcal{F}_A \bigotimes_x \overline{\ketbra{V_x}{V_x}\otimes\ketbra{V_x}{V_x}} \right].
\end{align}
The average over tensors can now be done explicitly,
\begin{align}
\overline{\ketbra{V_x}{V_x}\otimes\ketbra{V_x}{V_x}} &= \int dU \,\, (U\otimes U) \ketbra{0}{0} \otimes \ketbra{0}{0} (U^\dagger \otimes U^\dagger) \nonumber \\
&= \frac{\mathbb{I}_x+\mathcal{F}_x}{D_x^2+D_x}.
\end{align}
Inserting this into \ref{eq:barone} and using the result to take the average of equation \ref{eq:bartwo}, one finds that evaluating the 2nd Renyi entropy becomes equivalent to calculating the partition function of an Ising model with spin 1/2 variables. In the limit of large bond dimension $D_{xy}$ this partition function can be evaluated by approximating it by its minimal energy configuration.

Recall our notation in which $H_{A\bar{A}}$ is the boundary Hilbert space, with $H_A$ the Hilbert space enclosed by some cut $\gamma$ which we are considering. The bulk legs enclosed by $\gamma$ form the Hilbert space $H_{C}$, and those outside the cut form $H_{\bar{C}}$. Thus the entire network can be thought of as a map $M: C\bar{C}\rightarrow A\bar{A}$ or, by attaching maximally entangled pairs to the bulk legs, as a state $\ket{\Psi_M}_{A\bar{A}C\bar{C}}$. 

We would like to understand when a random tensor network will have the isometric subregion property, which we saw in section \ref{sec:subregionisometry} leads to both the RT and error correction properties. To do this, it is first useful to understand when the mapping $M: C\bar{C}\rightarrow A\bar{A}$ from the full bulk Hilbert space to the boundary is an isometry. This mapping being an isometry is actually equivalent to the state $\ket{\Psi_M}_{A\bar{A}C\bar{C}}$ being maximally entangled across the $C\bar{C}$ and $A\bar{A}$ subsystems. The random matrix techniques outlined above can be used to calculate the entropy $S(C\bar{C})$. Doing so, it is found that a necessary and sufficient condition for $M: C\bar{C}\rightarrow A\bar{A}$ to be an isometry is that 
\begin{align}\label{eq:dimensioncondition}
|\Omega| \log D_b < |\partial \Omega| \log D
\end{align}
where $\Omega$ is any region in the bulk, $D_b$ is the dimension of the bulk legs, and $D$ is the dimension of the planar legs. 

We now have the background we need to understand the isometric subregion property in random tensor networks. We will again look at entanglement properties of $\ket{\Psi_M}_{A\bar{A}C\bar{C}}$, in particular, we can show that if 
\begin{align}\label{eq:errorcorrectingconditions}
I(C:\bar{A}\bar{C})&=0 \nonumber \\
&\text{and} \nonumber \\
S(C) &= \log \dim C
\end{align}
then it follows that the map given by cutting the network along $\gamma$ is an isometry from the cut legs and bulk legs to the boundary. Indeed, the random matrix techniques can be used to show the above statements are true assuming \ref{eq:dimensioncondition} and taking a large bond dimension limit.

To see that \ref{eq:errorcorrectingconditions} implies the subregion isometry property, note that $I(C:\bar{A}\bar{C})=0$ gives that $\rho_{C\bar{C}\bar{A}} = \rho_A \otimes \rho_{\bar{A}\bar{C}}$ and that $S(C) = \log \dim C$ implies $\rho_A$ is maximally mixed. Thus
\begin{align}
\rho_{C\bar{C}A} = \frac{\mathbb{I}_C}{\dim C} \otimes \rho_{\bar{C}A}.
\end{align}
One purification of this state is $\ket{\Psi_M}_{C\bar{C}A\bar{A}}$, but another purification is
\begin{align}
\ket{\Psi} = (\ket{\Psi^+}^{\otimes n})_{CE} \otimes \ket{\phi}_{\bar{C}\bar{A}\bar{E}}
\end{align}
where $\ket{\Psi^+}$ are maximally entangled pairs. 

\begin{figure}
\centering
\begin{subfigure}{0.45\textwidth}
  \centering
   \includegraphics[scale=0.4]{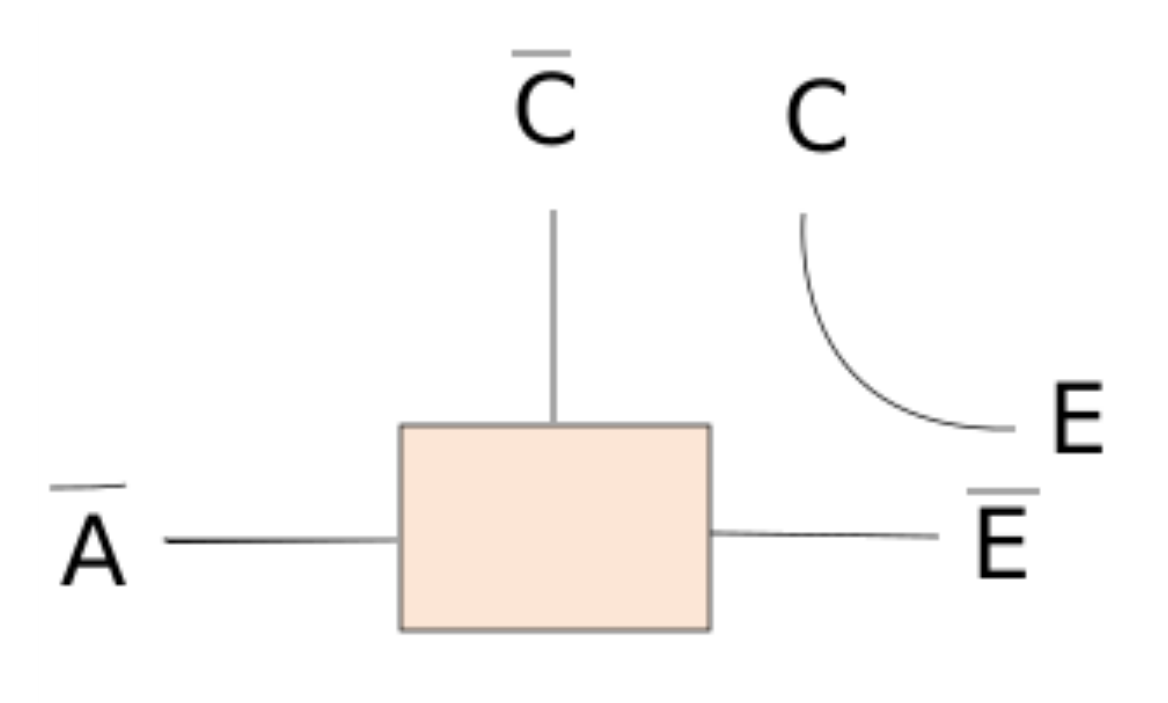}
\caption{}
\end{subfigure}
\hfill
\begin{subfigure}{0.45\textwidth}
  \centering
  \includegraphics[scale=0.4]{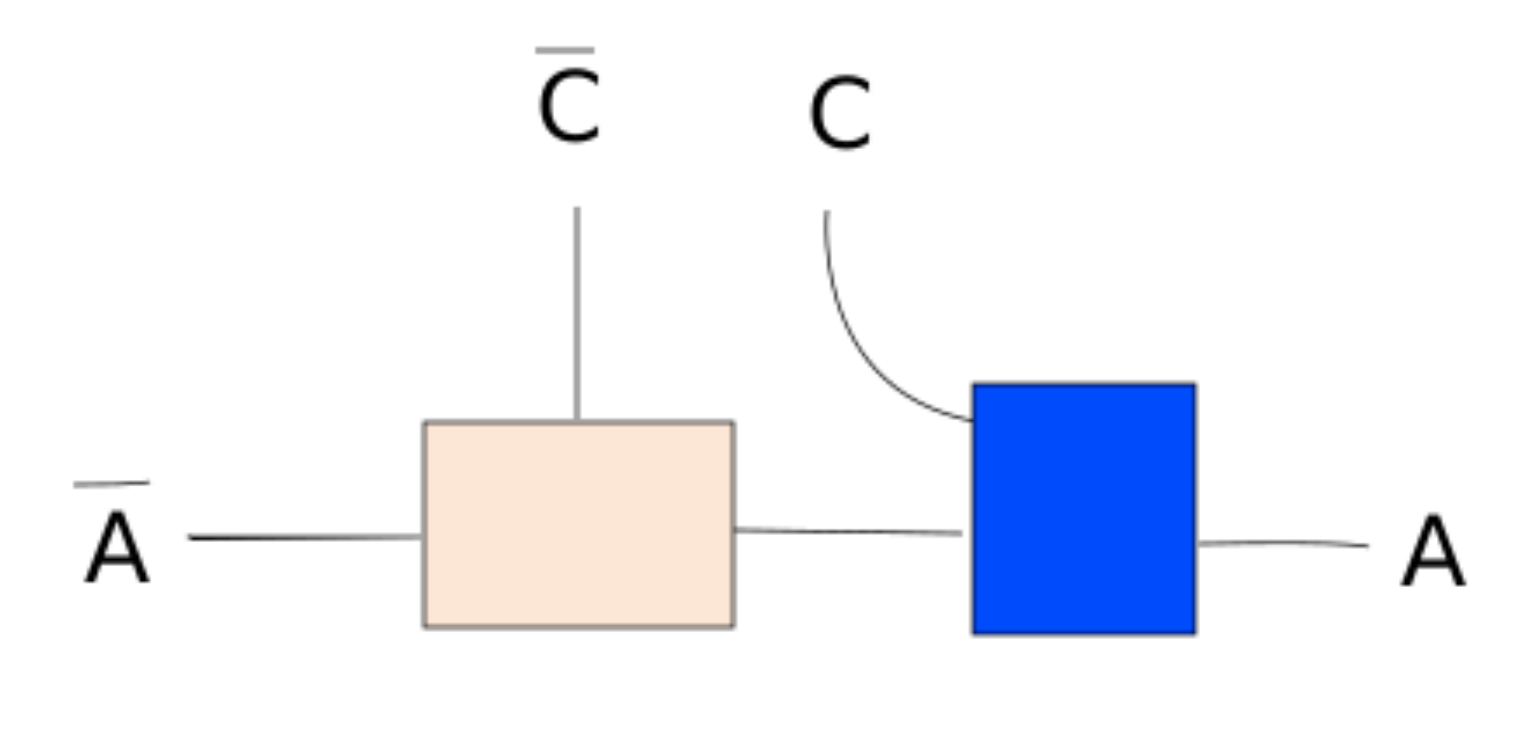}
\caption{}
\end{subfigure}
\caption{a) Structure of the state $\ket{\Psi}$. The beige box at left represents the state $\ket{\phi}_{\bar{A}\bar{C}\bar{E}}$. The curved black line represents the maximally entangled state $(\ket{\Psi^+}^{\otimes n})_{CE}$. b) Structure of the state $\ket{\Psi_M}$. The blue box represents the isometry $V:E \bar{E}\rightarrow A$ discussed in text.}
\label{fig:randomstate}
\end{figure}

Now, we make use of the fact that different purifications of the same state are related by isometries, in this case there must exist a map $V:\bar{E}E \rightarrow A$ such that
\begin{align}
\ket{\Psi_M} = V \ket{\Psi} = V((\ket{\Psi^+}^{\otimes n})_{CE} \otimes \ket{\phi}_{\bar{C}\bar{A}\bar{E}}). 
\end{align}
Figure \ref{fig:randomstate} illustrates the structure of this state diagrammatically. We see that the map $V$ can be identified with the map defined by $\gamma$ from a set of interior legs, which here form the system $\bar{E}$, and the bulk legs $C$ into the boundary $A$. Similarly, if we can establish that $I(\bar{C}:AC)=0$ then it follows that the map defined by the exterior of $\gamma$ is also an isometry from cut legs and bulk legs to the boundary region $\bar{A}$. Once we have that the two sides of the cut $\gamma$ are isometries, the error correction property and RT formula follow immediately by the same arguments as given in section \ref{sec:subregionisometry}. 

\chapter{Tensor networks for dynamic spacetimes}\label{sec:tndynamics}

The HaPPY networks of section \ref{sec:happynetworksreview} display the RT formula. In this sense these networks are a toy model with features analogous to AdS/CFT. In this analogy, the boundary legs play the role of CFT degrees of freedom and the tensor networks graph geometry plays the role of the geometry of a bulk Cauchy slice. It is interesting to understand the limitations of this toy model. In particular, we know that the RT formula applies only to static spacetimes. For a dynamic spacetime RT is replaced by the HRT or maximin formula, which so far our tensor network model contains no analogue of. 

In this chapter we will extend the HaPPY network models to include a network analogue of the maximin formula. Performing this extension forces on us a changed perspective, in particular highlighting \ref{eq:newRT} as the most productive way to understand the RT formula and connecting the tensor network picture more closely with the error correction picture developed by Harlow \cite{harlow2017ryu}. 

We will argue in section \ref{sec:maximinlessons} that the definition of length used previously in discussions of tensor networks, equation \ref{eq:graphlength}, is incompatible with any description of a dynamic spacetime. From there, we go on to develop a new definition of length in networks based on the notion of a state on a cut developed in section \ref{sec:basics}, and to discuss our tensor network model for a dynamic spacetime in \ref{sec:tndynamics}. 

\section{Length and extremal curves in tensor networks} \label{sec:graphlengthanddynamics}

\subsection{Lessons from the maximin formula} \label{sec:maximinlessons}

Recall that the maximin formula states that the entropy of a boundary region $A$ can be calculated as
\begin{align} \label{eq:maximin}
S(A) = \underset{\Sigma}{\text{max}} \left(\underset{\gamma_A}{\min}\,\, L(\gamma_A)  \right).
\end{align}
That is consider a spacelike slice of the boundary and a subset $A$ of this slice. On each spacelike surface $\Sigma$ which has the chosen boundary, calculate the length of the minimal surface homologous to $A$. From the set of all those lengths choose the largest element. The resulting length will give the entropy $S(A)$. 

To translate this statement into tensor network language it is clear that we need to think of a boundary state as associated with a set of networks. Indeed, as mentioned preceding equation \ref{eq:slickpsifree}, there are various ways we can modify a network while preserving the boundary state. Beginning with a defining network these transformations give a set of networks corresponding to a single boundary state. We will explore the possibility that a set of networks generated in this way can be searched over to calculate boundary entropies, analogous to optimization over spacelike slices in the maximin formula. 

However, suppose that we have decided on a set of networks and that the maximin formula is true for these networks and their boundary state. Then the maximization step of the maximin formula gives that the minimal lengths in each network are bounded above by the entropy,
\begin{align} \label{eq:lengthentropy}
\underset{\gamma_A}{\text{min}} \, L(\gamma_A) \leq S(A).
\end{align}
This is a key inequality restricting the possible definitions of the length $L(\gamma_A)$ in the network. Indeed, suppose that we took $L(\gamma)=L_G(\gamma)$, the graph length. Then the rank bound on the entanglement entropy given in \ref{eq:rankbound} says that
\begin{align}
S(A) \leq \underset{\gamma_A}{\text{min}}\, (\log(\dim \gamma_A)) = \underset{\gamma_A}{\text{min}} \, L_G(\gamma).
\end{align}
This is the opposite inequality to \ref{eq:lengthentropy}, so we have that $S(A) = \underset{\gamma_A}{\text{min}} \, L_G(\gamma)$ for all networks in the set optimized over. This means that every network in the set must contain the extremal curve anchored on $A$. Repeating this for each of the possible boundary regions, we would conclude that every network in the set must contain the extremal curves for each boundary region. In a dynamic spacetime however no one slice should contain all of the extremal curves. Taking the graph length then prevents any description of the geometry of slices other than the constant time slices of static spacetimes.

We see that a requirement for describing spacelike slices of dynamic geometries using tensor networks is a new definition of length in networks. With a definition of length in hand, one approach is to determine the extremal cut by variation over the set of networks. However, it turns out to be simpler to generalize the isometric cut formula \ref{eq:newRT} than to try and generalize the statement \ref{eq:RT} of RT in terms of minimal cuts. Indeed, a possible generalization of \ref{eq:newRT} is just \ref{eq:newRT} again, with the modification that the isometric cut can now be chosen from within a set of networks. We find in the next section that there is a simple way to define $L(\gamma)$ that has reasonable geometric properties and which extends the isometric cut formula to the dynamic setting.

\subsection{A definition of length in tensor networks} \label{sec:lengthdefined}

From our analysis of the maximin formula we know that the definition of length will need to be changed. We claim that there is a simple way to define $L(\gamma)$ which guarantees $S(A) = L(\gamma_{iso}^A)$ whenever such an isometric cut exists in the set of networks associated with the boundary state. To see this first calculate the boundary state on $A$ in terms of the operators defined by an isometric cut $\gamma_{iso}^A$. This is most easily done by looking again at figure \ref{fig:cutmapprocedure} and considering an arrow reversed duplicate of the network in figure \ref{fig:cutmapprocedure}a. We then contract the $\bar{A}$ legs and use that $D^\dagger D = \mathbb{I}$, which yields
\begin{align}
\rho_A = C\, \tr_{\bar{B}} (\ketbra{\Psi}{\Psi}) C^\dagger.
\end{align}
Since $C$ is an isometry, we have that $S(\rho_A) = S(\tr_{\bar{B}}\ketbra{\Psi}{\Psi})$. Next, consider the length $L(\gamma_{iso}^A)$. As discussed in section \ref{sec:stateoncut} any cut has a state associated with it, given by \ref{eq:stateoncut}. In particular since $\gamma_{iso}^A$ is an isometric cut we have
\begin{align} \label{eq:isostate}
\ket{\gamma_{iso}}= \ket{\Psi}_{B \bar{B}}.
\end{align}
From this it is clear that defining the length as the entropy of one side of $\ket{\Psi}$ would correctly compute $S(A)$. We prefer to write this more symmetrically as the mutual information 
\begin{align} \label{eq:mutualdef}
L(\gamma_{iso}^A) = \frac{1}{2}I_{\ket{\gamma}}(B:\bar{B}).
\end{align}
For a HaPPY network, the state $\ket{\Psi}$ is a product of maximally entangled pairs and the length of a minimal cut reduces to the graph length. What about an arbitrary cut in a HaPPY network? In this case we make use of the fact that when $\ket{\Psi}$ is product, the length of a minimal cut becomes a sum over each leg in the cut
\begin{align} 
L(\gamma_{iso}^A) = \sum_i \frac{1}{2}I_{\ket{\gamma}}(B_i:\bar{B}_i).
\end{align}
Thus it is natural to associate a length to each leg individually,
\begin{align} \label{eq:indivlength}
L(\gamma_i)=\frac{1}{2}I_{\ket{\gamma}}(B_i:\bar{B}_i).
\end{align}
For an arbitrary curve we can define its length to be the sum of the lengths of each leg, where we calculate the length of a single leg by finding an isometric cut containing that leg. In the HaPPY network this assigns all legs a length of $\log \dim \gamma_i$.

For a non-HaPPY network we can attempt to assign lengths to every leg by the same procedure of looking at the product factors of the isometric cuts which are in that network. In general however not every leg will be part of an isometric cut, and it may not be possible to assign a length to every leg. This manner of building up the lengths of arbitrary cuts using the lengths of isometric cuts is reminiscent of the differential entropy formula \cite{balasubramanian2014bulk, headrick2014holographic}. Extremal cuts in the differential entropy formula play the role of isometric cuts in the procedure outlined here. This is consistent with our interpretation of \ref{eq:newRT} as applying to dynamic spacetimes, since the isometric cuts of \ref{eq:newRT} are playing the role of extremal curves.

As a basic check on this definition of length, we should confirm that all isometric cuts passing a single leg will assign the same value of length to that leg. Indeed, for any isometric cut crossing a segment $\gamma_i$ the length of that segment is given by the mutual information $I(B:\bar{B})/2$ computed in the projecting state $\ket{\Psi_i}$, and is independent of the operators $C$ and $\bar{C}$ defined by whichever isometric cut has been chosen. Further, the projecting state $\ket{\Psi_i}$ is fixed for a given network. 

We can also notice that contracting the network in a different basis has no effect on the length of a cut. This follows from the invariance of the mutual information under local unitaries. Finally, it would be nice to see that if a cut $\gamma$ is composed of two segments $\gamma_1$ and $\gamma_2$ then 
\begin{align}\label{eq:addit}
L(\gamma) = L(\gamma_1) + L(\gamma_2).
\end{align}
We'll refer to this as the additivity property. The additivity property is not guaranteed by our definition of length, but rather depends on the structure of the state $\ket{\Psi}$ appearing in \ref{eq:isostate}. For example if this state is product across each leg, that is if
\begin{align}
\ket{\gamma_{iso}} = \bigotimes_{i=1}^n \ket{\Psi_i},
\end{align}
then the length is additive at the level of individual legs. However, in other cases it may happen that entanglement is present across the $B_i$, in which case the length will not be additive across individual legs. In the dynamic example given in section \ref{sec:dynamicexample} we will allow legs which are contracted with a common vertex to share entanglement, meaning additivity may fail at the level of a small number (in the case there, three) legs.


\subsection{A static example} \label{sec:staticexample}

As an illustration of this assignment of length we look at a network which does not satisfy RT when the graph length is used, but does when using the mutual information based definition. Our example is based on the network shown in figure \ref{fig:sixlegexample}a. The six legged tensors are perfect tensors, and the two legged tensors shown as solid black dots are maximally entangled pairs used to form the contraction. This is a HaPPY network and the boundary entropies are all given by the minimal number of legs cut to separate off a boundary region. 

\begin{figure}
\begin{center}
\includegraphics[scale=0.8]{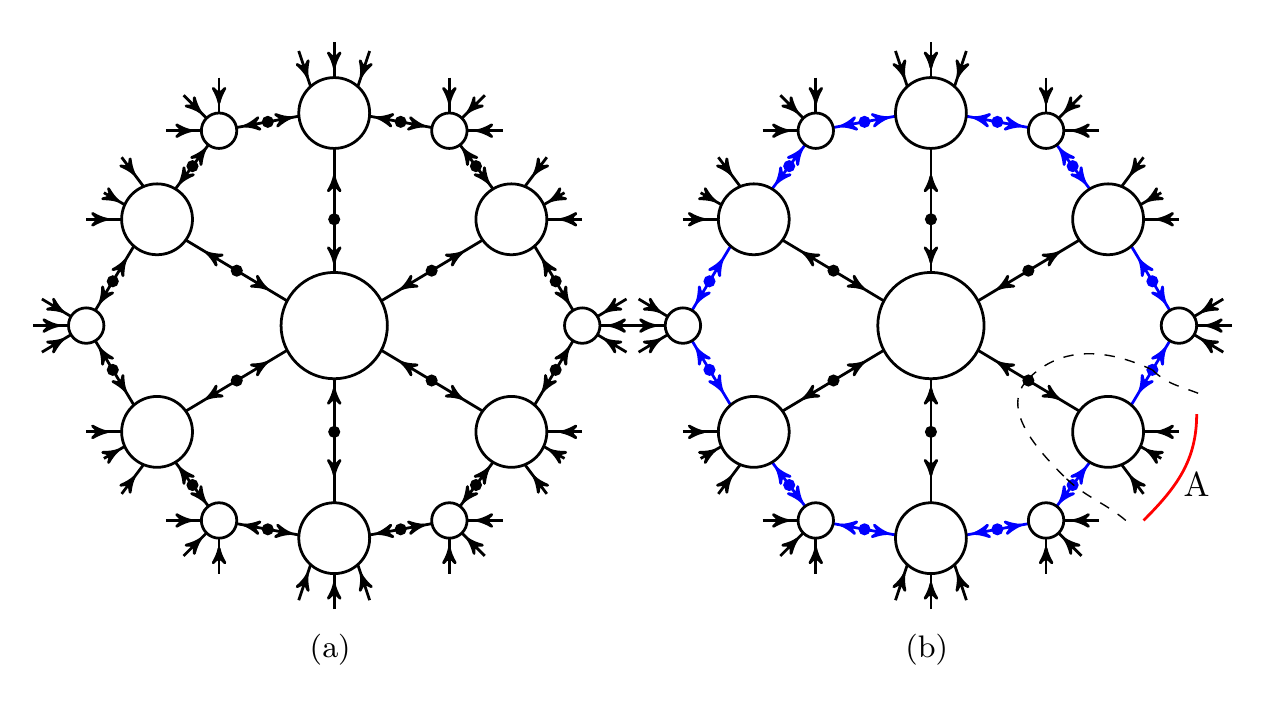}
\end{center}
\caption{(a) A HaPPY network. Six legged vertices are perfect tensors and two legged projecting pairs are maximally entangled states. All boundary entropies are given by the graph length of a minimal cut in the network. (b) A network which satisfies the Ryu-Takayanagi formula using the mutual information based definition of length, but does not satisfy Ryu-Takayanagi when using the graph length. The blue dots and legs represent the projecting state given in eq. \ref{eq:newprojector}. The dashed line represents the isometric cut for the boundary region $A$.}
\label{fig:sixlegexample}
\end{figure}

For our example we replace the maximally entangled pairs around the edge of the network with another state, $\ket{\Psi_i}$, which for convenience we write in the form\footnote{It is always possible to write $\ket{\Psi}$ in this way because we can write $\ket{\Psi_i} = A \otimes B \ket{\Psi^+}$
and then use the transpose rule to move $B$ to the other subspace.}
\begin{align} \label{eq:newprojector}
\ket{\Psi_i}=(\mathcal{O} \otimes \mathbb{I}) \ket{\Psi^+}.
\end{align}
The modified network is shown in figure \ref{fig:sixlegexample}b. We claim that this network satisfies the RT formula using our new definition of length in terms of mutual information, but not using the graph length. To see this we begin by computing the entropy of the three boundary legs marked region A.

\begin{figure}
\begin{center}
\begin{tikzpicture}[scale=0.16];
\draw[thick,postaction={on each segment={mid arrow}}] (0,0) -- (0,18);
\draw[thick,postaction={on each segment={mid arrow}}] (0,0) to [out=60,in=-60] (0,18);
\draw[thick,postaction={on each segment={mid arrow}}] (0,0) to [out=120,in=-120] (0,18);
\draw[thick,postaction={on each segment={mid arrow}}] (23,0) to [out=45,in=-45] (23,18);
\draw[thick,postaction={on each segment={mid arrow}}] (23,0) to [out=75,in=-75] (23,18);
\draw[thick,postaction={on each segment={mid arrow}}] (23,0) to [out=105,in=-105] (23,18);
\draw[thick,postaction={on each segment={mid arrow}}] (23,0) to [out=135,in=-135] (23,18);
\bluebrasegment{0}{18}{23}{18}{0.75}
\blueketsegment{0}{0}{23}{0}{0.75}
\arrowedsegment{0}{-3}{0}{-6}
\arrowedsegment{23}{-2}{23}{-6}
\arrowedsegment{0}{24}{0}{21}
\arrowedsegment{23}{24}{23}{20}
\arrowedsegment{-3}{0}{-6}{0}
\arrowedsegment{-6}{18}{-3}{18}
\arrowedsegment{20}{0}{24}{0}
\arrowedsegment{24}{18}{20}{18}
\arrowedsegment{29}{18}{25}{18}
\arrowedsegment{25}{0}{29}{0}
\draw[fill=white] (0,0) circle [radius=3];
\draw[fill=white] (0,18) circle [radius =3];
\draw[fill=white] (23,0) circle [radius=2];
\draw[fill=white] (23,18) circle [radius=2];
\node at (33,9) {$=$};

\draw[thick,postaction={on each segment={mid arrow}}] (40,18) to [out=0,in=0] (40,0);
\draw[thick,postaction={on each segment={mid arrow}}] (48,0) -- (48,18);

\draw[fill=blue] (56,0) circle [radius=0.75];
\draw[fill=blue] (56,18) circle [radius=0.75];

\draw[blue, thick, postaction={on each segment={mid arrow}}] (56,18) to [out=0,in=0] (56,0);
\draw[blue, thick, postaction={on each segment={mid arrow}}] (56,0) to [out=180,in=180] (56,18);
\draw[thick, postaction={on each segment={mid arrow}}] (69,0) to [out=180,in=180] (69,18);

\draw[dashed,thick] (69,18) -- (72,18);
\draw[dashed,thick] (69,0) -- (72,0);

\draw[dashed,thick] (40,18) -- (37,18);
\draw[dashed,thick] (40,0) -- (37,0);

\end{tikzpicture}
\end{center}
\caption{The basic simplification used to compute the density matrix of region A in figure \ref{fig:sixlegexample}b. The six leg tensors are from the edge of the network shown in figure \ref{fig:sixlegexample}b. The effect of the state given in equation \ref{eq:newprojector} is to add a normalization factor, represented as the blue loop at right.}
\label{fig:densityAcontraction}
\end{figure}
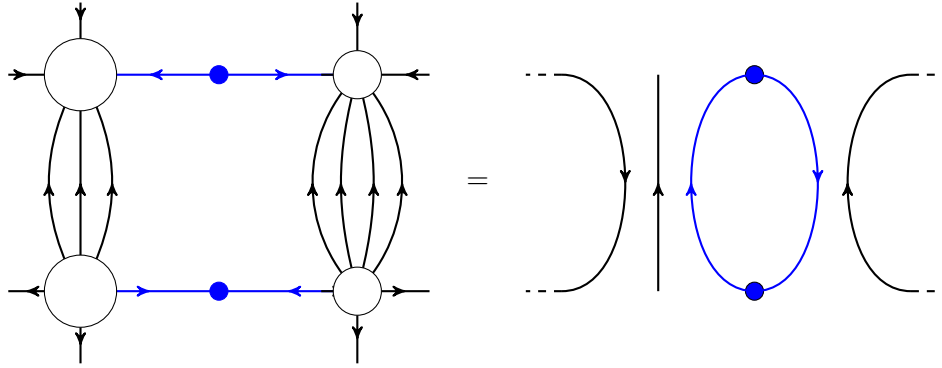

We can go a long ways towards computing this entropy using the graphical notation. To do this we draw an arrow reversed copy of figure \ref{fig:sixlegexample}b, then contract all the legs in $\bar{A}$. To understand what happens when this is done consider the diagram in figure \ref{fig:densityAcontraction}. The simplification shown there gives that all the insertions of $\mathcal{O}$ that are not adjacent to region $A$ turn into pure normalization factors. After continuing the contraction we are left with the density matrix illustrated in figure \ref{fig:densityAsixleg}. As an operator expression, this is
\begin{align} \label{eq:rhoa}
\rho_A = T (\mathcal{O}\mathcal{O}^\dagger\otimes \mathbb{I}\otimes\mathcal{O}\mathcal{O}^\dagger) T^\dagger.
\end{align}
The entropy is given by:
\begin{align} 
S(\rho_A) = S(T^\dagger \rho_A T) = 2\cdot S(\mathcal{O}\mathcal{O}^\dagger) + 1
\end{align}
By choosing $\mathcal{O}$ to be non-unitary we find an entropy less than $3=L_G(\gamma_{min})$, so we have that the RT formula using the graph length fails. 

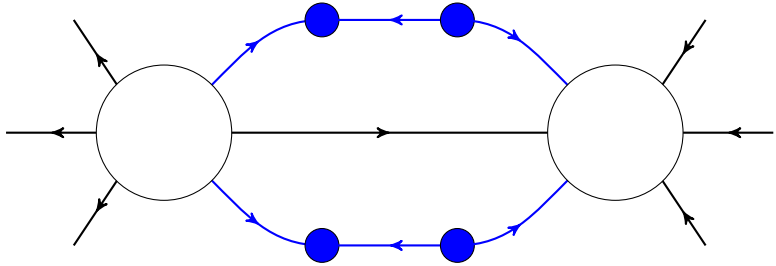
\begin{figure}
\begin{center}
\begin{tikzpicture}[scale=0.3]
\draw[thick,blue,postaction={on each segment={mid arrow}}] (2,2) to [out=45,in=180] (7,5);
\draw[fill=blue] (7,5) circle [radius=0.75];
\draw[thick,blue,postaction={on each segment={mid arrow}}] (13,5) to [out=180,in=0] (7,5);
\draw[fill=blue] (13,5) circle [radius=0.75];
\draw[thick,blue,postaction={on each segment={mid arrow}}] (13,5) to [out=0,in=135] (18,2);
\draw[thick,blue,postaction={on each segment={mid arrow}}] (2,-2) to [out=-45,in=180] (7,-5);
\draw[fill=blue] (7,-5) circle [radius=0.75];
\draw[thick,blue,postaction={on each segment={mid arrow}}] (13,-5) to [out=180,in=0] (7,-5);
\draw[fill=blue] (13,-5) circle [radius=0.75];
\draw[thick,blue,postaction={on each segment={mid arrow}}] (13,-5) to [out=0,in=225] (18,-2);
\draw[thick,postaction={on each segment={mid arrow}}] (0,0) -- (20,0);
\draw[thick,postaction={on each segment={mid arrow}}] (-3,0) -- (-7,0);
\draw[thick,postaction={on each segment={mid arrow}}] (-2,2) -- (-4,5);
\draw[thick,postaction={on each segment={mid arrow}}] (-2,-2) -- (-4,-5);
\draw[thick,postaction={on each segment={mid arrow}}] (27,0) -- (23,0);
\draw[thick,postaction={on each segment={mid arrow}}] (24,5)--(22,2);
\draw[thick,postaction={on each segment={mid arrow}}] (24,-5)--(22,-2);
\draw[fill=white] (0,0) circle [radius=3];
\draw[fill=white] (20,0) circle [radius=3];
\end{tikzpicture}
\end{center}
\caption{Graphical representation of the density matrix of the region $A$ from figure \ref{fig:sixlegexample} b. }
\label{fig:densityAsixleg}
\end{figure}

What is the minimal length when computed using eq. \ref{eq:mutualdef}? The cut with the minimal number of legs actually defines an isometry on both sides\footnote{To be precise, the $\bar{A}$ side of this cut has the property $\bar{C}^\dagger \bar{C} = \alpha \mathbb{I}$ for a scalar $\alpha$. This scalar is divided out when the normalization is added to the state.}. This is in fact what we used when showing that the reduced density matrix $\rho_A$ was given by expression \ref{eq:rhoa}. Since both sides of the cut are isometries the state on the cut is just given by the projecting state, which in this case is
\begin{align}
\ket{\Psi} = \ket{\Psi_i}_{B_1 \bar{B}_1} \otimes \ket{\Psi^+}_{B_2 \bar{B}_2} \otimes \ket{\Psi_i}_{B_3 \bar{B}_3}.
\end{align}
There are two legs with operator insertions, which have length
\begin{align}
L = \frac{1}{2}I_{\Psi_i}(B:\bar{B})= S(\mathcal{O}\mathcal{O}^\dagger),
\end{align}
while the leg with no insertion has length $I_{\Psi^+}(B:\bar{B})/2=1$, giving $L = 2\cdot S(\mathcal{O}\mathcal{O}^\dagger)+1 = S(\rho_A)$. It is straightforward to check the minimal lengths and boundary entropies of any other region $A$ in the network shown in figure \ref{fig:sixlegexample} b agree.

\begin{figure}
\begin{center}
\includegraphics[scale=0.7]{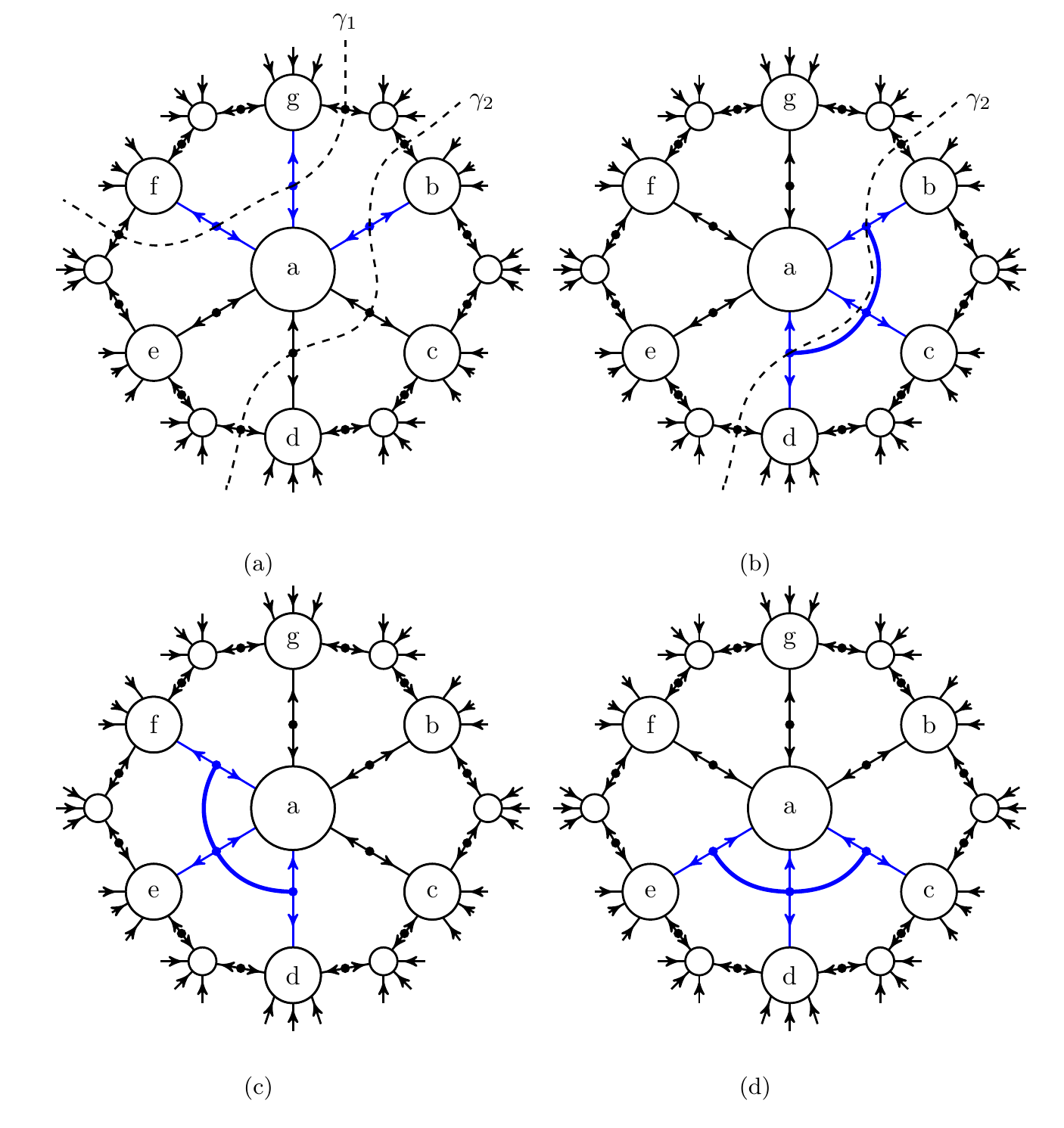}
\end{center}
\caption{The four networks included in the set $F$ considered in text. The projecting states shown in blue in (a) are defined by an operator $\mathcal{O}$ acting on the maximally entangled state, these can be pushed to any subset of three legs, as seen in (b-d). In general the pushed through operator will not have a tensor product structure, so the corresponding projecting state will be entangled across three legs. This is indicated here by the thick blue line in (b-d). All four networks shown here have the same boundary state. The entropies of subsets of boundary legs is calculated by choosing the network which contains an isometric cut enclosing those boundary legs, and applying the isometric cut formula. No one network contains an isometric cut for every boundary region, but the set of four networks together do.}
\label{fig:dynamicexample}
\end{figure}

\section{Dynamic tensor network states}\label{sec:dynamicexample}

Consider an AdS spacetime and a spacelike slice of its boundary. The maximin formula shows that the entanglement entropy of a given boundary region can be found by determining the area of an extremal surface extending into the AdS spacetime. For a dynamic spacetime these extremal surfaces may lie in many different slices of the interior. In the tensor network picture we have identified isometric cuts as the network analogue of extremal surfaces. Further, we have suggested a set of networks contracting to a single boundary state is the analogue of the set of spacelike slices of the bulk spacetime. A tensor network state which is analogous to an evolving spacetime then should have isometric cuts for different boundary regions living in different networks drawn from this set. 

We will call this set of networks $F$, and specify the networks it contains by giving a defining network $N_0$ along with a set of allowed transformation rules. Continuing our analogy, we view these transformations as corresponding to deformations of the interior spacelike slices. Importantly, these transformations must preserve the boundary state. An example of such a transformation was given as equation \ref{eq:slickpsifree}.

To construct examples of boundary states with a geometry corresponding to a dynamic spacetime we begin with the example network of figure \ref{fig:sixlegexample}a and replace three of the maximally entangled pairs which are projected into the central vertex with the state
\begin{align}
\ket{\Psi_i}=(\mathcal{O} \otimes \mathbb{I}) \ket{\Psi^+}.
\end{align}
This is our defining network, shown in figure \ref{fig:dynamicexample}a. The allowed transformations we take to be the operator pushing operation discussed in section \ref{sec:basics}. This results in the four networks shown in figure \ref{fig:dynamicexample}b-d being included in the set $F$.

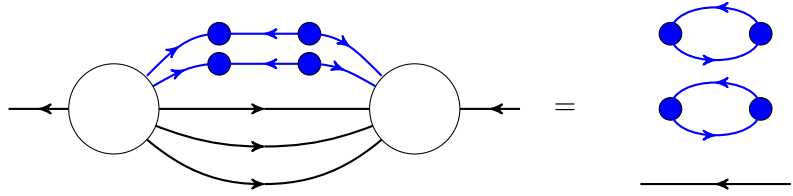
\begin{figure}
\begin{center}
\begin{tikzpicture}[scale=0.2]
\draw[thick,postaction={on each segment={mid arrow}}] (-3,0) -- (-7,0);
\draw[thick,postaction={on each segment={mid arrow}}] (27,0) -- (23,0);

\draw[thick,blue,postaction={on each segment={mid arrow}}] (2.59,1.5) to [out=30,in=180] (7,3);
\draw[fill=blue] (7,3) circle [radius=0.75];
\draw[thick,blue,postaction={on each segment={mid arrow}}] (13,3)--(7,3);
\draw[fill=blue] (13,3) circle [radius=0.75];
\draw[thick,blue,postaction={on each segment={mid arrow}}] (13,3)to [out=0,in=150] (17.4,1.5);

\draw[thick,->] (0,0) to [out=-22.5,in=180] (10,-2.5);
\draw[thick,->] (10,-2.5) to [out=0,in=-157.5] (20,0);
\draw[thick,->] (0,0) to [out=0,in=180] (10,0);
\draw[thick,->] (10,0) to [out=0,in=180] (20,0);

\draw[thick,blue,postaction={on each segment={mid arrow}}] (2.2,2.2) to [out=45,in=180] (7,5);
\draw[fill=blue] (7,5) circle [radius=0.75];
\draw[thick,blue,postaction={on each segment={mid arrow}}] (13,5)--(7,5);
\draw[fill=blue] (13,5) circle [radius=0.75];
\draw[thick,blue,postaction={on each segment={mid arrow}}] (13,5)to [out=0,in=135] (17.8,2.2);
\draw[thick,->] (0,0) to [out=-45,in=180] (10,-5);
\draw[thick,->] (10,-5) to [out=0,in=-135] (20,0);
\draw[fill=white] (0,0) circle [radius=3];
\draw[fill=white] (20,0) circle [radius=3];
\draw node at (30,0) {$=$};
\draw[thick,postaction={on each segment={mid arrow}}] (45,-5) -- (35,-5);

\draw[fill=blue] (37,5) circle [radius=0.75];
\draw[fill=blue] (43,5) circle [radius=0.75];
\draw[thick,blue,postaction={on each segment={mid arrow}}] (43,5) to [out=90,in=90] (37,5);
\draw[thick,blue,postaction={on each segment={mid arrow}}] (37,5) to [out=-90,in=-90] (43,5);

\draw[fill=blue] (37,0) circle [radius=0.75];
\draw[fill=blue] (43,0) circle [radius=0.75];
\draw[thick,blue,postaction={on each segment={mid arrow}}] (43,0) to [out=90,in=90] (37,0);
\draw[thick,blue,postaction={on each segment={mid arrow}}] (37,0) to [out=-90,in=-90] (43,0);

\end{tikzpicture}
\end{center}
\caption{The identity used to show a cut containing one interior leg which is blue in \ref{fig:dynamicexample} a is isometric up to a normalization factor. This identity is easily derived from that in figure \ref{fig:subperfect}.}
\label{fig:subperfect2}
\end{figure}

\begin{figure}
\begin{center}
\begin{tikzpicture}[scale=0.2]
\draw[thick,postaction={on each segment={mid arrow}}] (-2.2,-2.2) -- (-5,-5);
\draw[thick,postaction={on each segment={mid arrow}}] (-2.2,2.2) -- (-5,5);
\draw[thick,postaction={on each segment={mid arrow}}] (25,-5)--(22.2,-2.2);
\draw[thick,postaction={on each segment={mid arrow}}] (25,5) -- (22.2,2.2);
\draw[thick,->] (0,0) to [out=22.5,in=180] (10,2.5);
\draw[thick,->] (10,2.5) to [out=0,in=157.5] (20,0);
\draw[thick,->] (0,0) to [out=-22.5,in=180] (10,-2.5);
\draw[thick,->] (10,-2.5) to [out=0,in=-157.5] (20,0);
\draw[thick,blue,postaction={on each segment={mid arrow}}] (2.2,2.2) to [out=45,in=180] (7,5);
\draw[fill=blue] (7,5) circle [radius=0.75];
\draw[thick,blue,postaction={on each segment={mid arrow}}] (13,5)--(7,5);
\draw[fill=blue] (13,5) circle [radius=0.75];
\draw[thick,blue,postaction={on each segment={mid arrow}}] (13,5)to [out=0,in=135] (17.8,2.2);
\draw[thick,->] (0,0) to [out=-45,in=180] (10,-5);
\draw[thick,->] (10,-5) to [out=0,in=-135] (20,0);
\draw[fill=white] (0,0) circle [radius=3];
\draw[fill=white] (20,0) circle [radius=3];
\draw node at (30,0) {$=$};
\draw[thick,postaction={on each segment={mid arrow}}] (45,-5) -- (35,-5);
\draw[thick,postaction={on each segment={mid arrow}}] (45,0) -- (35,0);
\draw[fill=blue] (37,5) circle [radius=0.75];
\draw[fill=blue] (43,5) circle [radius=0.75];
\draw[thick,blue,postaction={on each segment={mid arrow}}] (43,5) to [out=90,in=90] (37,5);
\draw[thick,blue,postaction={on each segment={mid arrow}}] (37,5) to [out=-90,in=-90] (43,5);
\end{tikzpicture}
\end{center}
\caption{The identity used to show the cut $\gamma_1$ in figure \ref{fig:dynamicexample}a is isometric up to a normalization factor.}
\label{fig:subperfect}
\end{figure}
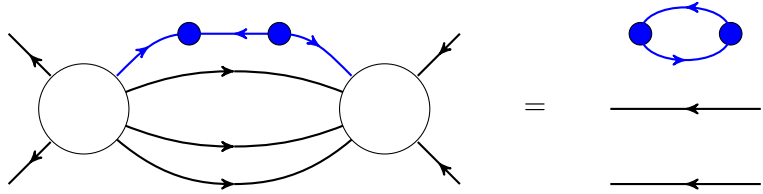

It is not difficult to see that no one of these networks contains isometric cuts for every boundary region, but the set of four together do. Consider for example the defining network, figure \ref{fig:dynamicexample} a, with $\mathcal{O}$ sitting on three of the interior legs (marked as blue legs). Consider first the subregion consisting of exterior legs attached to vertex f. The cut enclosing f and crossing three legs is isometric, this follows from the identity shown in figure \ref{fig:subperfect2}. Similarly, the cut $\gamma_2$ which encloses f and g can be shown to be isometric by use of the identity in figure \ref{fig:subperfect}. Further, the boundary legs adjacent to vertices f,g,b or c,d,e are also enclosed by isometric cuts contained in the network of figure \ref{fig:dynamicexample}a, since a minimal cut which crosses three interior blue legs is isometric. 

The need for additional networks to be included in the set $F$ arises when we consider subregions adjacent to two or fewer black interior legs. Take for example the subregion containing vertices b, c and d. The cut $\gamma_2$ which crosses one blue and two black interior legs is not isometric, nor is the other possible cut which crosses one black and two blue interior legs. To find an isometric cut enclosing b-c-d we consider network \ref{fig:dynamicexample} b. Since all the operator insertions now live on the cut $\gamma_2$ and $\gamma_2$ crosses a minimal number of legs, it is isometric. Similarly, an isometric cut for enclosing d-e-f can be found in network \ref{fig:dynamicexample} and an isometric cut for e-f-g in network \ref{fig:dynamicexample}. 

A complication arises in considering the region containing only c, d or e, or regions f-e, e-d, d-c, c-b. Consider region d-c, the remaining possibilities are handled similarly. In this case an isometric cut can be found in network \ref{fig:dynamicexample} b. To see this, we must return to the notion of a cut through the network. Recall that a cut $\gamma$ corresponds to a specification of projecting state, on which operators $C$ and $\bar{C}$ act to prepare the boundary state. Implicitly, choosing a cut involves specifying which legs of the projecting state are acted on by the operator $C$ and which by the operator $\bar{C}$, corresponding to our breakdown of the projecting state Hilbert space into $\mathcal{H}_B$ and $\mathcal{H}_{\bar{B}}$. To specify this in the graphical notation we can use a double line to specify a cut through the network. One cut $\gamma$ crosses the legs which are associated with the $B$ Hilbert space, with a second cut $\bar{\gamma}$ denoting the legs in the $\bar{B}$ Hilbert space. We have adopted this notation in figure \ref{fig:doubleline} to specify an isometric cut for the region c-d. It is straightforward to show that the operators defined by this cut are isometric, from which we can conclude that the mutual information across the $B$ and $\bar{B}$ subsystems of the projecting state is equal to the entropy of the boundary region c-d, as needed. 

\begin{figure}
\begin{center}
\includegraphics[scale=0.6]{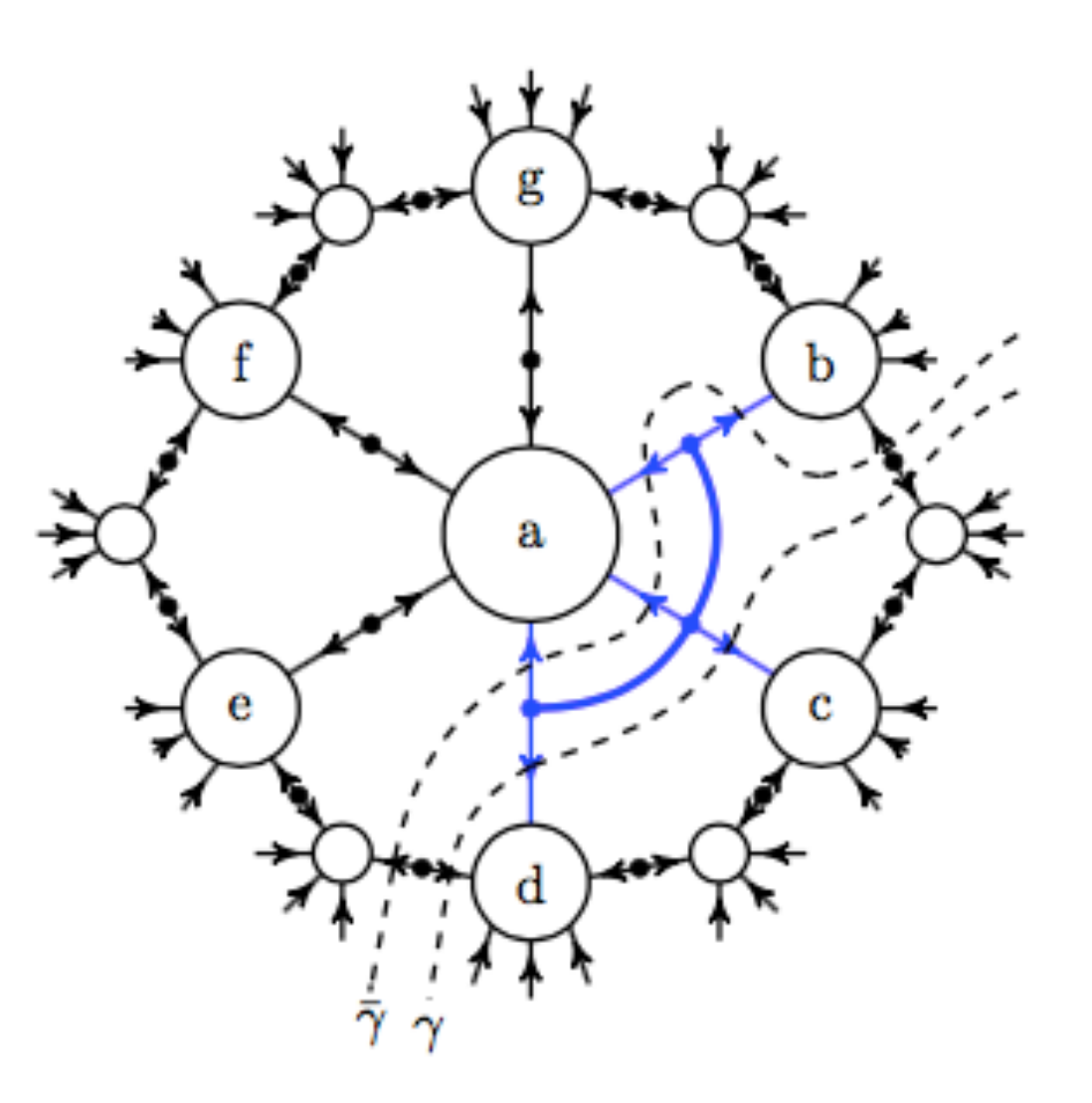}
\end{center}
\caption{Illustration of how to choose an isometric cut for the subregion consisting of legs adjacent to the d and c vertices. The upper dashed line crosses legs included in the $\bar{B}$ Hilbert space, while the lower dashed line crosses legs included in the $B$ Hilbert space. It is straightforward to check that both the operators above and below the dashed lines are isometries; it follows that $S(A) = \frac{1}{2}I(\bar{B}:B)$.}
\label{fig:doubleline}
\end{figure}

One advantage to the isometric cut formula is that it is not necessary to limit the networks which are included in the set $F$ searched over. Indeed, it is a consequence of our definitions that any cut which is isometric will have as its length the entropy of the enclosed boundary region. Thus although we specified only the four networks given in figure \ref{fig:dynamicexample} we could add arbitrary networks to the set $F$. Those without isometric cuts would not disturb the isometric cut formula at all, and any additional networks having isometric cuts would give unchanged values for the boundary entropies. 

It is not difficult to construct further examples of sets of networks satisfying the isometric cut formula based on the construction used here. Indeed, we may continue the pattern of contraction given by figure \ref{fig:sixlegexample}a and construct networks with an arbitrary number of layers. We can then proceed to replace a subset of the projecting maximally entangled pairs by non-maximally entangled states; allowing the same freedom of pushing operators through adjacent tensors then gives a set of networks associated with a fixed boundary state. We have not yet systematically studied which networks defined in this way contain isometric cuts for all possible boundary regions; doing so remains a direction for future work.

It is interesting to reconsider the dynamic picture developed earlier in the presence of bulk legs. If we assume we have a tensor network with the isometric subregion property, a nice picture emerges. We can consider evolving the bulk state forward in time by applying an operator to a portion of its Hilbert space. When we map this bulk state to the boundary by projecting it into the tensor network, we can push through this time evolution operator to the interior legs of the tensor network. The time evolution operator can then be treated as the same type of deformations to the network as were considered in this section, showing that local time evolution leads to deforming the Cauchy slice. Since we do not have examples of HaPPY networks which have the subregion isometry property, one could instead use the random tensor networks of section \ref{sec:randomreview}. 

\chapter{Holographic tensor networks away from AdS/CFT} \label{sec:networksfromgeometry}

In this chapter we consider starting with a continuous geometry and then building a network which approximates it. This allows us to construct networks which approximate spacetimes of interest other than AdS, for instance flat space. It is possible to do this by making use of the random tensor construction.

\section{Building a tensor network for flat space}\label{sec:flatspace}

In the random tensor construction it is possible to prove the RT formula when their are no bulk legs very generally \cite{hayden2016holographic}. In particular, there is no condition analogous to \ref{eq:dimensioncondition} when no bulk legs are present - a network constructed on any graph will display the RT formula. However, we should keep in mind that as with the other entropy calculations performed using random tensor techniques the RT formula is proven for random tensors in a limit of large bond dimension.

Since RT holds without restrictions on the graph geometry, we can try to construct networks that approximate flat space. This was already claimed in ref. \cite{hayden2016holographic}, however, we argue that the construction given there is problematic and give an alternative construction. We also wish to acknowledge that our construction borrows a technique from Bao et al. \cite{bao2015holographic}. 

We consider a disk $M = \{ (x,y) : 0 \leq x^2+y^2 \leq 1\}$ which is endowed with a distance function $d(u,v)$. Our goal is to fill in the disk with a planar tensor network which satisfies the Ryu-Takayanagi formula and whose minimal surfaces have lengths approximating the function $d(u,v)$. We introduce a parameter $\delta$ which represents a unit of length in the continuous geometry. In particular we say the network approximates the geometry of the disk to a resolution of $\epsilon$ if
\begin{align} \label{eq:approx}
\left|d(u,v) - \frac{\delta}{\log D} \cdot L(u,v)\right| \leq \epsilon,
\end{align}
where $L(u,v)$ is the graph length in the network, $D$ is the dimension of the legs in the network. The $\delta / \log{D}$ should be understood as a conversion factor from graph length (unitless) to physical length, so that $\delta$ is the length associated with cutting one leg of dimension $D$. Below we construct a network and show that, in this network, given any choice of resolution $\epsilon$ there is a choice of $\delta$ such that \ref{eq:approx} is satisfied.

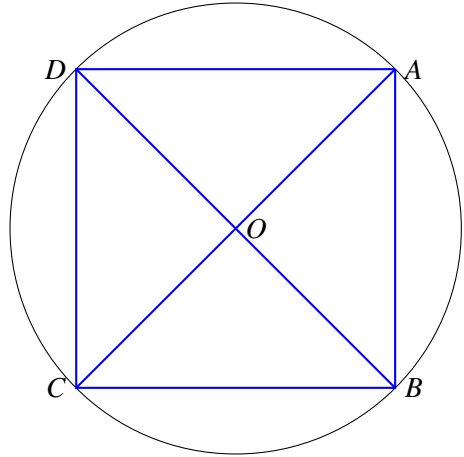
\begin{figure}
\begin{center}
\begin{tikzpicture}[scale=0.2]
\draw (0,0) circle [radius=15]; 
\draw[blue,thick] (10.6,10.6) -- (10.6,-10.6) -- (-10.6,-10.6) --(-10.6,10.6) --(10.6,10.6);
\draw[blue,thick] (10.6,10.6) -- (-10.6,-10.6);
\draw[blue,thick] (10.6,-10.6) -- (-10.6,10.6);
\node[right] at (0,0) {$O$};
\node[right] at (10.6,10.6) {$A$};
\node[right] at (10.6,-10.6) {$B$};
\node[left] at (-10.6,-10.6) {$C$};
\node[left] at (-10.6,10.6) {$D$};
\end{tikzpicture}
\end{center}
\caption{Illustration of our procedure for constructing a network whose minimal lengths approximate those of a given geometry. The example shown constructs a network with four boundary legs which approximates a disk shaped region of $\mathbb{R}^2$. In the first step, four boundary points are chosen and all of the minimal cuts anchored on those points are drawn. The minimal cuts form a planar graph, in this example the graph has vertices $A,B,C,D$ and $O$.}
\label{fig:stepone}
\end{figure}

Our strategy is to construct a graph whose minimal cuts satisfy \ref{eq:approx} and then populate that graph with random tensors of large bond dimension. With random tensors placed on the vertices, the results of ref. \cite{hayden2016holographic} then guarantee the Ryu-Takayanagi formula is satisfied. Let us consider as an example a disk which is a section of flat space, so $d(u,v) = \sqrt{(u_1-v_1)^2+(u_2-v_2)^2}$. A reasonable first approach is to tile the disk with a regular polygon. This can be done with triangles, squares, or hexagons. However, none of these tilings correctly reproduce lengths in the disk in the sense of \ref{eq:approx}. For example in a tiling with squares, the graph length function is 
\begin{align}
L(u,v) = |u_1 - v_1| + |u_2 - v_2|,
\end{align}
which doesn't approximate the Euclidean distance. Additionally, such a distance function gives highly degenerate minimal surfaces - for example a staircase shaped path gives the same distance between two boundary points as a path which turns only once. Regular tilings using triangles or hexagons produce similar graph distance functions and also have degenerate minimal surfaces.

\begin{figure}
\begin{center}
\begin{tikzpicture}[scale=0.2]
\draw (0,0) circle [radius=15]; 
\draw[dashed, blue, thin] (10.6,10.6) -- (10.6,-10.6) -- (-10.6,-10.6) --(-10.6,10.6) --(10.6,10.6);
\draw[dashed,blue,thin] (10.6,10.6) -- (-10.6,-10.6);
\draw[dashed,blue,thin] (10.6,-10.6) -- (-10.6,10.6);
\node[right] at (0,0) {$O$};
\node[right] at (10.6,10.6) {$A$};
\node[right] at (10.6,-10.6) {$B$};
\node[left] at (-10.6,-10.6) {$C$};
\node[left] at (-10.6,10.6) {$D$};

\draw[red,thick] (6,0) -- (18,0);
\draw[red,thick] (-6,0) -- (-18,0);
\draw[red,thick] (0,6) -- (0,18);
\draw[red,thick] (0,-6) -- (0,-18);
\draw[red,thick] (0,6) -- (6,0) -- (0,-6) -- (-6,0) -- (0,6);

\node[above right] at (0,6) {S};
\node[below right] at (6,0) {T};
\node[below right] at (0,-6) {U};
\node[below left] at (-6,0) {V};

\node[above] at (0,18) {W};
\node[below] at (0,-18) {X};
\node[right] at (18,0) {Y};
\node[left] at (-18,0) {Z};

\end{tikzpicture}
\end{center}
\caption{Illustration of the second step in our procedure for constructing a network which approximates a given geometry. In this step, the dual of the graph formed in step one is drawn. The edges of the dual graph are assigned a weight based on the $\mathbb{R}^2$ length of the edges in the direct graph they cut. For example, a weight of $\text{Floor}( \bar{AB} / \delta)$ is assigned to the edge $TY$, where $\delta x$ is a parameter with units of length controlling how closely the graph approximates lengths in the disk.}
\label{fig:steptwo}
\end{figure}
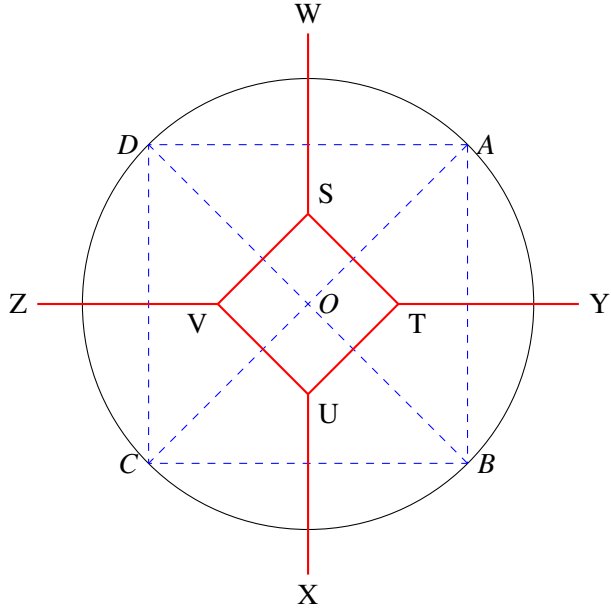

Our construction begins by specifying a set of points on the edge of the disk. The closeness of our approximation is set in part by the number of points on the boundary chosen, which we will denote by $N$. The construction proceeds by drawing every minimal surface between pairs of these points; this is illustrated in figure \ref{fig:stepone}. The resulting surfaces define a graph which we take to be the dual graph of the tensor network being constructed. Importantly, the edges in the direct graph are assigned a weighting $w_i$ set by
\begin{align}
w_i = \text{Floor}(d(u_j,u_k) / \delta).
\end{align}
Up to the rounding implemented by the floor function, the weight of the edges in the direct graph is given by the length of the edges in the dual graph which they cut, measured in units of $\delta$. Forming the direct graph from the dual graph and assigning the weightings is illustrated in figure \ref{fig:steptwo}.  

Finally, random tensors are placed on the vertices of the direct graph and the number of legs along an edge is chosen to be equal to the weighting $w_i$ associated with that leg. We can then show that the resulting tensor network has lengths which satisfy \ref{eq:approx}. To prove this, note that a cut in the tensor network is also a path in the dual graph. We will consider a minimal cut passing from $u_0\rightarrow u_N$ where the $u_i$ are vertices in the dual graph. Consider one segment of that path which passes from $u_i$ to $u_j$. The length of this segment is given by 
\begin{align}
L(u_i,u_j) = (\text{number of legs crossed})\log D.
\end{align}
The number of legs crossed is just $w_i$, which is the length of that segment of the path given in units of $\delta$, 
\begin{align}
L(u_i,u_j) = \text{Floor}\left(\frac{d(u_i,u_j)}{\delta}\right) \cdot \log D
\end{align}
Inserting this into \ref{eq:approx} gives that
\begin{align} 
\left|d(u_i,u_j) - \frac{\delta}{\log D} \cdot L(u_i,u_j)\right| \leq \delta.
\end{align}
The number of boundary points chosen, $N$, sets the maximal number of segments in a minimal cut through the disk, which we call $f(N)$\footnote{A chord $AB$ divides the points $C,D...$ on the circles edge into two sets of size $n_1$ and $n_2$ where $n_1+n_2\leq N-2$. Since every pairing of such points gives a chord which crosses $AB$ once we can bound the number of cuts through $AB$ by $((N-2)/2)^2$}. Then the triangle inequality gives that for a minimal path through the network
\begin{align}
d(u_0,u_n) \leq \delta f(N).
\end{align}
A network with resolution $\epsilon$ then can be constructed by choosing $\delta = \epsilon / f(N)$. 

It would be interesting to explore further the properties of the flat space network constructed here, and to understand if other flat space network constructions are possible. For example, the network here is far from being translationally invariant, and it is interesting to ask if a translationally invariant flat space network can be constructed. This seems unlikely, as translational invariance requires we use one of the regular tilings with triangles, squares, or hexagons, which produce non-unique minimal surfaces and distance measures which do not approximate flat space. One possibility however is to use a random tiling. In this case translational invaraince is restored at large enough distance scales, but minimal surfaces in the random tiling have some hope of being unique. 

It would also be interesting to understand if the flat space network constructed here, or a possible random tiling, could be upgraded to have bulk legs and to have the isometric subregion property. More generally, the construction here allows us to take any geometry and use its minimal surfaces to construct a corresponding tensor network. We can thus ask for any geometry about its properties as a mapping. 

\chapter{Final remarks}

One of the fundamental puzzles of quantum gravity concerns the exact relationship between geometry and entanglement. The tensor network has the basic advantage of giving an immediate relationship between these two apparently distinct ideas, at least at a discrete level. The work presented in this thesis regarding dynamic spacetimes strengthens this connection by showing tensor networks can capture not just the geometry of a special class of spacelike slices (constant time slices in static, asymptotically AdS spacetimes), but actually capture features of the geometry of arbitrary Cauchy slices in any asymptotically AdS spacetime.

From this perspective, a natural question poses itself: how can timelike directions be represented in a tensor network or similar formalism? Since in the dynamic tensor network picture it is the mutual information which naturally defines lengths, this seems connected to a standing question in quantum information theory of how to define a mutual information between a Hilbert space at a particular time and the same Hilbert space at a later time. 

Another direction the tensor network-geometry-entanglement connection might be pursued is towards a continuum picture. It would be interesting to understand how to take the continuum limit of a tensor network, or to use tensor network ideas to inspire the development of microscopic models that incorporate the entanglement-geometry connection naturally. Two active directions in this area are the development of continuous-MERA (cMERA) \cite{haegeman2013entanglement} and a new perspective on viewing Euclidean path integrals as continuum limits of tensor networks \cite{caputa2017liouville}.

As we began to address in section \ref{sec:networksfromgeometry}, the tensor network gives a way to think about entanglement and geometry outside of the context of AdS/CFT. As seen perhaps most explicitly in the random tensor models there can be bulk to boundary mappings where the bulk is not hyperbolic. This presents the possibility of toy models for flat space holography, a direction we feel has not yet been satisfactorily explored. It would be interesting to understand if the mapping from bulk to boundary defined by a flat space network can ever be an isometry. 

Finally, we mention a few other directions which are being pursued in recent holographic tensor network literature. First, an alternative approach to understanding the dynamics of tensor networks has appeared \cite{osborne2017dynamics}. The approach taken there focuses on extending the symmetries of a HaPPY network to include time translation. Nicely, they are able to show that the symmetry group of their HaPPY network forms a group which had been previously understood to approximate the conformal group. The approach given in this thesis has the advantage of displaying a network analogue of the maximin formula. It would be interesting to understand better the connection between these two approaches.

Suggestions have been made that complexity of a boundary CFT state be dual to the volume \cite{brown2015complexity} or action \cite{brown2016complexity} of a certain bulk region. This was initially in part inspired by the idea of a tensor network, which can also be understood as a quantum circuit, as building up spacetime. Indeed, the correct tensor network to describe a geometry may be related to the efficiency of that network in building the boundary state, as suggested by the MERA, which makes the complexity-volume proposal natural. A clear picture of how time evolution works in tensor networks would be useful in elucidating this potential connection.


\begin{singlespace}
\raggedright
\bibliographystyle{abbrvnat}
\bibliography{biblio}
\end{singlespace}

\appendix

\backmatter

\end{document}